\documentclass[reprint,pra,superscriptaddress,groupedaddress, twocolumn]{revtex4} 
\usepackage[utf8]{inputenc}
\usepackage{graphicx,amsmath,bm, color}
\usepackage{latexsym}
\usepackage{amssymb}
\usepackage{amsfonts}
\usepackage{amsthm}
\usepackage{thmtools}
\usepackage{mathrsfs}
\usepackage{mathtools}
\usepackage{natbib}
\usepackage{hyperref}
\usepackage{cleveref}
\usepackage{verbatim,graphics, float}
\usepackage{psfrag}
\usepackage{subfigure}
\usepackage{lipsum}
\usepackage{dsfont}
\usepackage{ulem} 

\newcommand{\ba}{\begin{eqnarray}}
\newcommand{\ea}{\end{eqnarray}}
\newcommand{\ban}{\begin{eqnarray*}}
\newcommand{\ean}{\end{eqnarray*}}

\newcommand{\ket}[1]{\mbox{$ | #1 \rangle $}}
\newcommand{\bra}[1]{\mbox{$ \langle #1 | $}}
\newcommand{\ketbra}[2]{\ket{#1} \! \bra{#2}}
\newcommand{\pure}[1]{\ketbra{#1}{#1}}
\newcommand{\trace}{\mathrm{Tr}}
\newcommand{\id}{\mathbb{I}} 
\newcommand{\pr}{\mathrm{P}} 
\newcommand{\projec}[1]{\mbox{$ | #1 \rangle\! \langle #1 |$}}
\newcommand{\proj}{P} 
\newcommand{\flip}[1]{\overline{#1}} 
\newcommand{\defvar}{\coloneqq} 
\newcommand{\entVN}{H}
\newcommand{\Bkey}{\ensuremath{\chat{\rv{B}}'_1}}
\newcommand{\Bkeyreg}{\ensuremath{\hat{\rv{B}}'_1}}
\newcommand{\rv}[1]{\textsf{{\itshape #1}}} 
\newcommand*{\chat}[1]{\hspace{0.15em}\hat{\phantom{#1}}\kern-0.8em #1}
\newcommand{\be}{\begin{equation}}
\newcommand{\ee}{\end{equation}}
\newtheorem{theorem}{Proposition}

\begin{document}
\title{Noisy pre-processing facilitating a photonic realisation of \\ device-independent quantum key distribution}


\date{\today}

\author{M. Ho}
\affiliation{Department of Physics, University of Basel, Klingelbergstrasse 82, 4056 Basel, Switzerland}
\affiliation{Department of Applied Physics, University of Geneva,
Chemin de Pinchat 22, 1211 Geneva, Switzerland}
\author{P. Sekatski}
\affiliation{Department of Physics, University of Basel, Klingelbergstrasse 82, 4056 Basel, Switzerland}
\author{E.Y.-Z. Tan}
\affiliation{Institute for Theoretical Physics, ETH Zürich, 8093 Zürich, Switzerland}
\author{R. Renner}
\affiliation{Institute for Theoretical Physics, ETH Zürich, 8093 Zürich, Switzerland}
\author{J.-D. Bancal}
\affiliation{Department of Physics, University of Basel, Klingelbergstrasse 82, 4056 Basel, Switzerland}
\affiliation{Department of Applied Physics, University of Geneva, Chemin de Pinchat 22, 1211 Geneva, Switzerland}
\author{N. Sangouard}
\affiliation{Department of Physics, University of Basel, Klingelbergstrasse 82, 4056 Basel, Switzerland}
\affiliation{Universit\'e Paris-Saclay, CEA, CNRS, Institut de physique th\'eorique, 91191, Gif-sur-Yvette, France}

\begin{abstract}

Device-independent quantum key distribution provides security even when the 
equipment used to communicate over the quantum channel is largely uncharacterized. 
An experimental demonstration of device-independent quantum key distribution is however challenging. A central obstacle in photonic implementations is that the global detection efficiency, i.e., the probability that the signals sent over the quantum channel are successfully received, must be above a certain threshold. We here propose a method to significantly relax this threshold, while maintaining provable device-independent security. This is achieved with a protocol that adds artificial noise, which cannot be known or controlled by an adversary, to the initial measurement data (the raw key). Focusing on a realistic photonic setup using a source based on spontaneous parametric down conversion, we give explicit bounds on the minimal required global detection efficiency.
\end{abstract}
\maketitle

\paragraph{Introduction.---} \emph{Quantum key distribution (QKD)} allows two parties,  Alice and Bob,  who are connected by a quantum channel, to generate a secret key~\cite{Bennett84,Ekert91}. QKD has been demonstrated in countless experiments, see Refs.~\cite{Gisin02, Scarani09, Lo14} for reviews. The security of QKD, i.e., the claim that an adversary, Eve, who may fully control the quantum channel, gains virtually no information about the key, usually relies on the assumptions that (i)~quantum theory is correct; (ii)~Alice and Bob can exchange classical messages authentically, i.e., an adversary cannot alter them; (iii)~Alice and Bob's devices are trusted, i.e., they carry out precisely the operations foreseen by the protocol~\cite{Renner08}. 


The last assumption is hard to meet in practice. This leads to vulnerabilities, as demonstrated by hacking experiments~\cite{Zhao08,Lydersen10,Weier11,Garcia19}. The aim of \emph{device-independent} QKD is to overcome this problem --- it provides security  even when assumption~(iii) is not satisfied. One usually distinguishes between different levels of device-independence, depending on what the assumption is replaced with. In \emph{measurement-device-independent} QKD, one considers prepare-and-measure protocols and drops assumption~(iii) for the measurement device but not for the preparation device~\cite{Braunstein12,Lo12,Rubenok13,Ferreira13,Liu13,Tang14,Comandar16}. One therefore still needs to trust the latter to generate precisely calibrated quantum states. This requirement is dropped in \emph{fully device-independent} QKD, which is the topic of this work.

Fully device-independent QKD protocols are entanglement-based~\cite{Ekert91}. A completely untrusted source distributes entangled signals to Alice and Bob who measure them. Alice and Bob's measurement device is each modelled as a \emph{black box}, which takes as input a choice of measurement basis and outputs the corresponding measurement outcome. Crucially, however, one does not need to assume that the black boxes actually carry out the intended measurement. Instead of assumption (iii), it is then sufficient to assume that any information that, according to the protocol, must be processed locally by Alice and Bob remains in their respective labs, that the inputs to the devices can be chosen independently from all other devices, and that their outputs are post-processed with trusted computers.

The generation of a secure key requires sufficiently many signal pairs that are sufficiently strongly entangled. The relevant measure of entanglement is the amount by which the measurement statistics violates a Bell inequality, such as the Clauser-Horne-Shimony-Holt (CHSH) inequality~\cite{CHSH69}. Intuitively, a significant violation of this inequality guarantees that Alice and Bob's state is close to a two-qubit maximally entangled state which cannot be shared with a third party~\cite{Kaniewski16, Valcarce2020}. This, in turn, guarantees that Eve's information about Alice and Bob's measurement outcomes is bounded. A mere violation of the CHSH inequality is however not sufficient for secure key distribution to be possible. Rather, the amount of violation must be above a certain threshold, which depends on the \emph{global detection efficiency}, i.e., the probability that Alice's measurement (the same applies to Bob's measurement) on a given pair of entangled signals is successful \footnote{Consider for example a two-qubit maximally entangled state measured with detection systems having a detection efficiency $\eta_A$ for Alice and $\eta_B$ for Bob. Further consider for simplicity the case where $\eta_A=\eta_B=\eta.$ If both measurements succeed, which happens with probability $\eta^2$, the CHSH value can be as high as $2\sqrt{2}$. If one of them fails, which happens with probability $2\eta(1-\eta),$ the results are completely uncorrelated and the CHSH score is $0$. When none of the measurement succeeds, which happens with probability $(1-\eta)^2,$ the results are classically correlated and the CHSH score is at most equal to $2.$ A violation of the CHSH inequality is only possible in this scenario if $\eta \geq \frac{2}{\sqrt{2}+1} \approx 82.8\%.$ Note that the threshold can be lowered to $\approx 67\%$ is obtained with non-maximally entangled two-qubit state~\cite{Brunner14}.}. For secure key distribution to be possible, it must be above a threshold, which depends on the protocol as well as on the security proofs~\cite{Pironio09,Vazirani14,ArnonFriedman18}. The best (minimum) threshold known so far is beyond the range achievable by state-of-the-art experiments. 

The main contribution of this Letter is to propose a protocol for fully device-independent QKD that has a significantly lower threshold for the global detection efficiency, and prove its security against general attacks. The protocol includes a step where artificial noise is added to the measurement outcomes --- a method that has been known to lead to improvements in conventional (device-dependent) quantum cryptography~\cite{Renner05a, Renner05b, Renes07}. 
The additional noise damages both the correlation between Alice and Bob and the correlation to Eve. But since the possibility to generate a key depends on the difference between the strengths of these correlations, the net effect can still be positive. As our calculations show, this is indeed the case.

For concreteness, we consider an implementation on an optics platform as shown in Fig.~\ref{Fig1}, where entangled photons are generated by spontaneous parametric down conversion (SPDC) and measured with photon counting techniques. Such a setup is appealing as it enables high repetition rates, so that a key can be generated after a reasonable running time. We will however need to take into account that the statistics of this photon source intrinsically limits the maximum CHSH violation~\cite{CapraraVivoli15}. In this context, we prove that noisy pre-processing provides a significant reduction of the requirement on the detection efficiency.\\

\begin{figure}
\begin{center}
\includegraphics[scale=0.4]{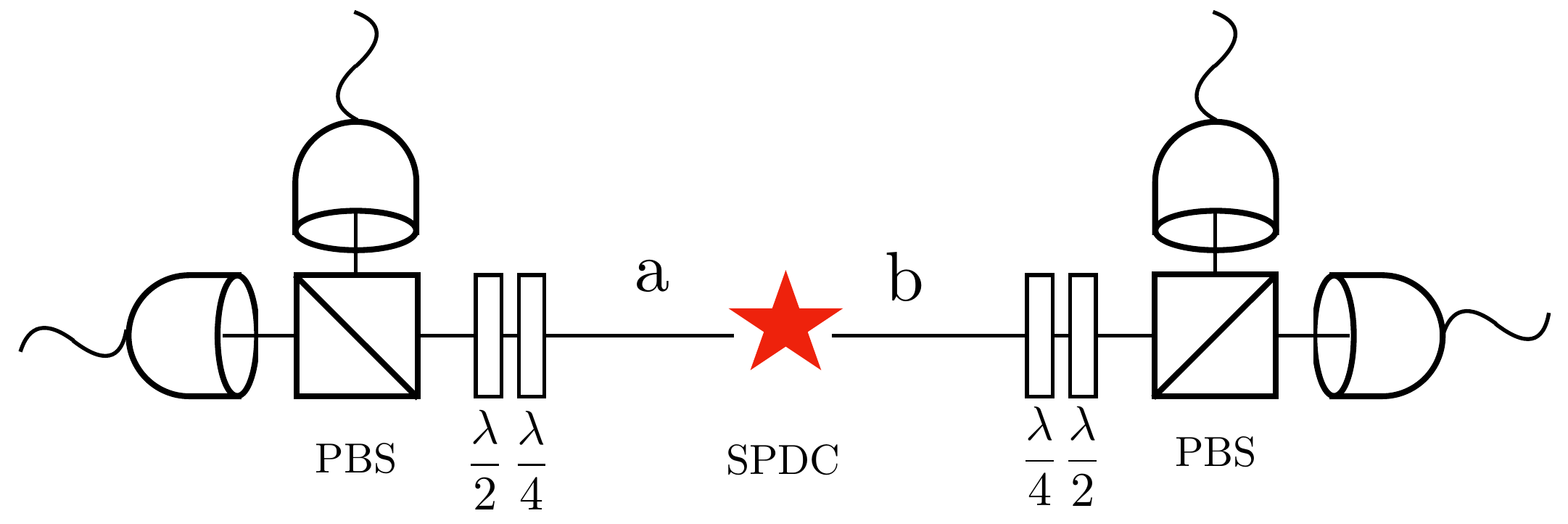}
\end{center}
\caption{Experimental platform envisioned for a device-independent implementation of quantum key distribution. A source (red star) based on spontaneous parametric down conversion (SPDC) is used to create photon pairs entangled in polarization. Alice receives mode $a$ while Bob receives mode $b$, and they perform measurements using a polarising beamsplitter (PBS) and two detectors. A set of wave-plates ($\frac{\lambda}{4}, \frac{\lambda}{2}$) allow them to choose the measurement setting.}
\label{Fig1}
\end{figure}

\paragraph{Protocol.---} For our device-independent QKD protocol, we consider a source that repeatedly distributes a pair of entangled signals (encoded in photonic modes) to Alice and Bob. Alice measures her part of each pair using a measurement $A_x$ with setting $x\in\{0,1,2\}$ chosen at random.
Similarly, Bob measures his part with a measurement $B_y$ where $y\in\{1,2\}$ is a random setting. 
While the measurement outcomes may in general be non-binary, we suppose that for $x, y \in \{1,2\}$ they are processed locally and turned into binary values $\rv{A}_x, \rv{B}_y \in \{-1,+1\}$ for Alice and Bob, respectively. In a \emph{parameter estimation step}, Alice and Bob communicate classically to estimate, using a sample of their results, the \emph{CHSH score} 
\begin{equation}
S=\langle A_1 B_1\rangle + \langle A_1 B_2\rangle + \langle A_2 B_1\rangle - \langle A_2 B_2\rangle
\end{equation}
 where 
 \begin{equation}
 \langle A_x B_y\rangle = p(\rv{A}_x=\rv{B}_y | x, y) - p(\rv{A}_x \neq \rv{B}_y | x, y) 
 \end{equation}
 quantifies the correlation between the outcomes for measurement choices $x$ and $y$, respectively. The measurement setting $x=0$ does not enter the score, but instead is chosen to generate an outcome $\rv{A}_0$ that minimises the uncertainty about $\rv{B}_1$ (quantified in terms of $H(\chat{\rv{B}}_1|\rv{A}_0)$, see Eq.~\eqref{keyrate} below). Bob then forms the \emph{raw key} from the outcomes $\rv{B}_1$ of the pairs that Alice measured with respect to $x=0$.

The next step of our protocol is known as \emph{noisy pre-processing}~\cite{Renner05a, Renner05b, Renes07}. Bob is instructed to generate new raw key bits $\hat{B}_1$ by flipping each of the bits $B_1$ of his initial raw key independently with probability~$p$. The protocol then proceeds with an error correction step that allows Alice to infer Bob's new (noisy) raw key. In a final privacy amplification step, Alice and Bob apply a hash function to this raw key to obtain the final secret key. \\


\paragraph{Key generation rate from CHSH score.---} Suppose that the source has emitted $n$ pairs of entangled signals, and denote their joint quantum state by $\Psi_{ABE}$. In the device-independent scenario, nothing is known about this state, nor the state space. Using the entropy accumulation theorem, one can show however that the entropic quantities that are relevant for measuring Eve's information are (up to terms that are sublinear in $n$) basically the same as the case where the devices are memoryless and behave identically and independently in each communication round of the protocol~\cite{ArnonFriedman18}. In particular, we may assume that $\Psi_{ABE} = \psi_{ABE}^{\otimes n}$ where $\psi_{ABE} \in \mathcal{H}_A \otimes \mathcal{H}_B \otimes \mathcal{H}_E$ is a tri-partite state compatible with the CHSH score. We note however that the dimension $d_{A/B}= \mathcal{H}_{A/B}$ of Alice and Bob's state space may still be arbitrary.  

In the asymptotic limit of large $n$, the key generation rate when optimal one-way error correction and privacy amplification is used is given by~\cite{DevetakWinter}
\begin{align} 
r_{} &=  \entVN(\chat{\rv{B}}_1|E) - H(\chat{\rv{B}}_1|{\rv{A}}_0), \label{keyrate}
\end{align}
where $\entVN$ is the von Neumann entropy (which includes the Shannon entropy as a special case). The first term, which quantifies Eve's uncertainty, can be expanded as
\begin{equation}\begin{split}
\entVN(\chat{\rv{B}}_1|E)
&=H(\chat{\rv{B}}_1) - \left(\entVN(\rho_E) - \sum_{b} p_{b} \entVN(\hat{\rho}_{E|b})\right),
\end{split}\label{eq:cont ent}
\end{equation}
where $\rho_E = \text{tr}_{AB}|\psi_{ABE} \rangle\langle \psi_{ABE}|$ is Eve's reduced state and $\hat{\rho}_{E|b}$ is Eve's state conditioned on the event that Bob's noisy key bit $\chat{\rv{B}}_1$ equals $b$, which occurs with probability $p_{b}$. By showing the equivalence of the protocol to one that includes a symmetrisation step (see Supplemental Material, Appendix A), where Bob flips the outcomes of each raw key bit $\rv{B}_1$ depending on a public random string, one can assume without loss of generality that $p_{b}=\frac{1}{2}$.

To bound $H(\rho_E)$ and $\entVN(\hat \rho_{E|b}),$ we use the approach presented in~\cite{Pironio09}, which we briefly outline here (see Supplemental Material, Appendix B). There, one first uses Jordan's lemma~\cite{Supic19} to choose a basis for $\mathcal{H}_A$ in which the observables $A_x$ for $x \in \{1,2\}$ are simultaneously block diagonal, with blocks of dimension at most 2 that we label $\lambda$. 
This can be expressed via a factorization of Alice's Hilbert space in the form $\mathcal{H}_A=\mathcal{H}_{A'}\otimes \mathcal{H}_{A''}$ so that her measurement operators are given by
\begin{align}
A_x = \sum_{\lambda} A_x^\lambda \otimes \pure{\lambda}
\end{align}
where $A_x^\lambda$ are qubit Pauli measurements and $\{\ket{\lambda}\}$ form an orthonormal basis for $A''$. Similarly, for Bob we can write $B_y = \sum_{\mu} B_y^\mu \otimes \pure{\mu}$. Since there are two possible measurements for Alice and Bob, the Pauli measurements $A_x^\lambda$ and $B_y^\mu $ can be taken to define the $X$-$Z$ planes of the Bloch spheres in each block. 

By introducing an ancilla $\mathcal{H}_{R''}$ that is a copy of $\mathcal{H}_{A''} \otimes \mathcal{H}_{B''}$, one can bound (see Supplemental Material, Appendix B) Eve's uncertainty on $\chat{\rv{B}}_1$ by instead bounding $\entVN(\chat{\rv{B}}_1|E R'')$ for states of the form 
\begin{align}
\sum_{\lambda\mu} p_{\lambda\mu} 
\pure{\lambda\mu\lambda\mu}_{A''B''R''}
\otimes 
\pure{\Psi^{\lambda\mu}}_{A'B'E}
\label{betterstate}
\end{align}
with  
\begin{equation}
\ket{\Psi^{\lambda\mu}}_{A'B'E} = \sum_{i=1}^4 \sqrt{L_i} \ket{\Phi_i}_{A'B'}\ket{i}_E
\end{equation}
where $\ket{\Phi_i}$, for $i=1,\dots,4$, form a Bell basis (with respect to the $X$-$Z$ plane defined by the measurements) and the weights $L_i$ can be taken to satisfy $L_1\geq L_2$, $L_3\geq L_4$\footnote{Each value of $(\lambda,\mu)$ potentially corresponds to a different Bell basis, different weights $L_i$ and orthonormal states $\ket{i}_E.$ This is not explicitly mentioned in the choice of notations to keep the latter as simple as possible.}. 
The block-diagonal structure of the state implies $\entVN(\chat{\rv{B}}_1|ER'')= \sum_{\lambda \mu} p_{\lambda\mu}  \entVN(\chat{\rv{B}}_1|E)_{\Psi^{\lambda\mu}}$, where $\entVN(\chat{\rv{B}}_1|E)_{\Psi^{\lambda\mu}}$ is the value produced by the state $\ket{\Psi^{\lambda\mu}}_{A'B'E}$ after measurement and noisy pre-processing.

We now evaluate Eq.~\eqref{eq:cont ent} for any single block defined by $\lambda$ and $\mu$. Let
$S^{\lambda\mu}$ be the CHSH score of the block. The term $H(\rho_E^{\lambda \mu})$ for this block is simply the entropy of the weight distribution $\mathbf{L} = (L_1, \ldots, L_4)$. The term $\entVN(\hat{\rho}_{E|b}^{\lambda \mu})$ is more complicated to compute as the conditional states $\hat{\rho}_{E|b}^{\lambda \mu}$ 
also depend on the angle $\phi$ of Bob's measurement $B_1(\phi)= \cos(\phi)\sigma_z + \sin(\phi)\sigma_x$, and the amount of noise added.  The effect of the latter is to mix the conditional states corresponding to the two possible outcomes $\pm 1$
\begin{equation}
    \hat{\rho}_{E|{b=\pm 1}}^{\lambda \mu} = (1-p) \rho_{E|{b=\pm 1}}^{\lambda \mu} + p\, \rho_{E|{b=\mp 1}}^{\lambda \mu},
\end{equation}
where $\rho_{E|{b=\pm 1}}^{\lambda \mu} = \textrm{tr}_{A'B'} \mathds{1}_{A'E} \otimes\left(\mathds{1} \pm B_1(\phi)\right) \projec{\Psi^{\lambda\mu}}_{A'B'E}$. In contrast to the conditional states $\rho_{E|{b=\pm 1}}^{\lambda \mu}$ without noise, the eigenvalues of the states $\hat{\rho}_{E|{b=\pm 1}}^{\lambda \mu}$ do not have a simple expression, hence making the computation of $\entVN(\hat{\rho}_{E|{b=\pm 1}}^{\lambda \mu})$ more difficult than in the absence of noisy pre-processing~\cite{Pironio09}. 

However, we can show that Eve's uncertainty $\entVN(\chat{\rv{B}}_1|E)_{\Psi^{\lambda\mu}}$ is an increasing function of the angle $\phi \in[0,\pi/2]$ of Bob's measurement, see Appendix C section 1. This allows us to conclude that Eve's information is maximized for $\phi=0$. In a next step, we show that $H(\chat{\rv B}_1|E)_{\Psi^{\lambda\mu}}$ is minimized for
$L_2=L_4=0$, $L_1= \frac{1}{4} \left(2+\sqrt{(S^{\lambda\mu})^2-4}\right)$ and $L_3=1-L_1$, see Supplemental Material, Appendix C section 2. This shows the state and measurement minimizing Eve's ignorance is independent of $p$, and hence identical to the one for the case where Bob does not introduce any artificial noise. The resulting bound on Eve's uncertainty is of the form $\entVN(\chat{\rv{B}}_1|E)_{\Psi^{\lambda\mu}} \geq 1- I_p(S^{\lambda\mu})$ with
\begin{equation}
\label{eq:noisybound}
\begin{split}
I_p(S)= &h\left(\frac{1+\sqrt{(S/2)^2-1}}{2} \right) \\
- &h \left( \frac{1+  \sqrt{1- p(1-p) (8-S^2)}}{2}  \right),
\end{split}
\end{equation}
where $h$ denotes the binary entropy, see Supplemental Material, Appendix C section 2. Combining the convexity of the function $1-I_p(S)$ with the relation $\entVN(\chat{\rv{B}}_1|ER'')= \sum_{\lambda \mu} p_{\lambda\mu}  \entVN(\chat{\rv{B}}_1|E)_{\Psi^{\lambda\mu}}$ for the state~\eqref{betterstate}, we deduce the overall bound
\begin{equation}
    \entVN(\chat{\rv{B}}_1|ER'') \geq 1- I_p(S).
\end{equation}
In particular, it follows from Eq.~\eqref{keyrate} that the secret key rate with optimal error correction satisfies
\begin{equation} \label{eq_fullkeyrate}
     r \geq 1-I_p(S) - H(\chat{\rv{B}}_1|\rv{A}_0).
\end{equation} \\

\paragraph{Optics implementation.---} It has been recently shown by multiple experiments that it is possible to violate a Bell inequality (without the need of a fair sampling assumption) with an optical implementation \cite{Giustina13, Christensen13, Shalm15, Giustina15, Shen18, Liu18}, similar to the one shown in Fig.~\ref{Fig1}. An SPDC source produces a bipartite state according to the Hamiltonian $H= i\sum_{k=1}^{N} (\chi a_k^\dagger b_{k,\perp}^\dagger - \bar{\chi} a_{k,\perp}^\dagger b_k^\dagger - h.c.  )$, where the bosonic operators $a_k, a_{k, \perp} (b_k, b_{k, \perp})$ refer to one of the $N$ modes received by Alice (Bob) with two possible polarizations labelled with and without the subscript $\perp$.
$\chi$ and $\bar{\chi}$ are related to the non-linear susceptibility and can be controlled by appropriately tuning the pump power and its polarization. The resulting state is given by 
\begin{equation}
\nonumber
|\psi\rangle = (1-T_g^2)^{N/2} (1-T_{\bar{g}}^2)^{N/2} \Pi_{k=1}^N e^{T_g a_k^\dagger b_{k,\perp}^\dagger - T_{\bar{g}} a_{k,\perp}^\dagger b_k^\dagger} |\underbar{0}\rangle
\end{equation}
where $T_g=\tanh g$ and $T_{\bar g}=\tanh \bar g$ with $g= \chi t $ and $\bar{g}= \bar{\chi} t$ squeezing parameters related to $a_k^\dagger b_{k,\perp}^\dagger$ and $a_{k,\perp}^\dagger b_k^\dagger$, respectively. Measurements are done with a polarization beamsplitter and two non-photon number resolving detectors. A set of wave-plates is finally used to choose the measurement setting locally. 

In such an optical implementation, each of the measurements applied by Alice and Bob has four possible outcomes: no-click on both detectors, one click in either one of the two detectors, and two clicks. In the case of $A_1$, $A_2$, $B_1$, and $B_2$, these results are binned to form the binary values required for the computation of $S$. For example, Alice may set $\rv{A}_1=1$ if she observed one click in one specific detector and none in the other, and $\rv{A}_1=-1$ for the other three variants. However, following the proposal presented in Ref.~\cite{Ma12}, no such binning is carried out for~$\rv{A}_0$, because this value is used only for error correction.

A central performance parameter for Bell-type experiments is the detection efficiency $\eta_A$ ($\eta_B$) , which is defined as the overall probability for a photon emitted from the source to be detected 
at Alice's (Bob's) location~\cite{Brunner14}. We assume that Alice's and Bob's detection efficiencies are equal; $\eta_A=\eta_B=\eta$. We now use Eq.~\eqref{eq_fullkeyrate} to determine the threshold for $\eta$ above which the secret key rate is positive. For this, we optimise the number of modes $N$ which affects the photon statistics, the squeezing parameters $g$ and $\bar g$ which change the average number of photons in each mode, the measurement settings, and the noise parameter~$p$, see Supplemental Material, Appendix D.\\

\paragraph{Results and comparison to other protocols.---} Our main finding, shown as a blue solid  line in Fig.~\ref{Fig2}, is that a positive key rate can be obtained with an SPDC source as soon as the detection efficiency is larger than $83.2\%$.

\begin{figure}
\begin{center}
\includegraphics[width=0.52\textwidth]{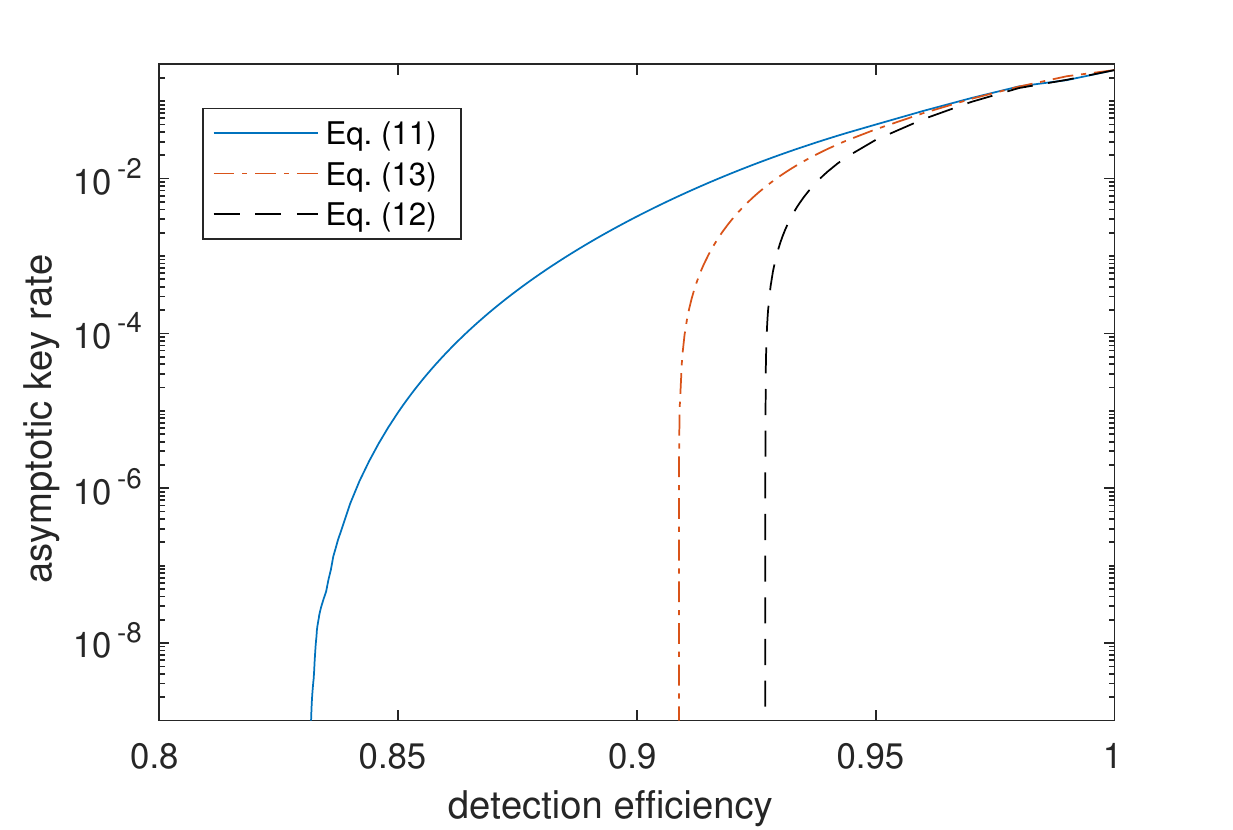}
\end{center}
\caption{Key rates vs system efficiency for the protocols based on a SPDC type source combined with photon counting techniques described in the main text. The key rate obtained from Ref.~\cite{Pironio09}, where the error correction cost is computed from the quantum bit error rate, shows positive key rates at $\eta \geq 92.7 \%$ (dashed black curve). The use of a refinement proposed in Ref.~\cite{Ma12} where the error correction cost is estimated directly from the measurement results, improves it to $ 90.9 \%$ (dashed-dotted red curve). Further employing noisy pre-processing brings the critical efficiency to $83.2\%$ (solid blue curve).}
\label{Fig2}
\end{figure}

For comparison, we consider the protocol studied in~\cite{Pironio09}. The difference to our protocol are two-fold. Firstly, no artificial noise is added, i.e., $p=0$. Secondly, $\rv{A}_0$ is taken to be a binary value. In this case, the bound on the rate reduces to 
\begin{equation}
r \geq 1-h\left(\frac{1+\sqrt{(S/2)^2-1}}{2} \right) -h(Q)
\end{equation}
where  $Q=p({\rv{A}_0 \neq \rv{B}_1})$ is the quantum bit error rate. The bound corresponds to the black dashed line in Fig.~\ref{Fig2}. The maximum key rate is below $1$ for unit detection efficiency due to the photon statistics. The  key rate decreases as the efficiency decreases and a minimum global detection efficiency of $\eta=92.7\%$ is needed to have a positive key rate.

If we consider instead  a protocol like the one in~\cite{Pironio09}, but with error correction that uses the four-valued outcome $\rv{A}_0$ (the four possible outcomes correspond physically to i) no click at all, ii) and iii) one click exactly in one of the two detectors and iv) a twofold coincidence) instead of its binarisation, one already obtains an improvement~\cite{Ma12}. The rate is now given by
\begin{equation}
r \geq 1-h\left(\frac{1+\sqrt{(S/2)^2-1}}{2} \right) -H(\rv{B}_1|\rv{A}_0),
\end{equation}
and represented as the red dashed-dotted line in Fig.~\ref{Fig2}. (See Supplemental Material, Appendix~A~for details regarding the computation of $H(\rv{B}_1|\rv{A}_0)$.)  The maximum rate is essentially unchanged but the behaviour as a function of the detection efficiency is different. In particular, the threshold for the global detection efficiency is $90.9\%$. Furthermore, when adding noisy pre-processing of the raw key into the protocol of Ref.~\cite{Pironio09} with error correction using the four valued outcome $\rv{A}_0$~\cite{Ma12}, the requirement on the detection efficiency goes down to 83.2\%.\\

\paragraph{Discussions and conclusion.---}
SPDC sources allowed several groups worldwide to test Bell inequalities without detection loopholes~\cite{Giustina13, Christensen13, Shalm15, Giustina15, Shen18, Liu18} and to conclude the existence of non-local causal correlations with a high statistical confidence. The observed violation of Bell inequalities in these experiments was however very limited due the intrinsic photon statistics and non-unit detection efficiency. One may thus wonder whether an extension of existing detection-loophole-free Bell tests to device-independent QKD is even possible. We computed the requirement on the detection efficiency using known security proofs. We then significantly reduced this requirement by deriving a new security proof for the case where noisy pre-processing is applied before further classical processing.

Noisy pre-processing of the raw key  was known to increase the resistance to noise of conventional QKD. For BB84~\cite{Bennett84} for example, the critical quantum bit error rate goes from 11\% without noisy pre-processing to 12.4\% with it -- a 13\% relative improvement~\cite{Renner05a, Renner05b, Renes07}. In the case of DIQKD, noisy pre-processing lowers the efficiency from 90.9\% to 83.2\% a relative tolerance improvement of 78\%. This might have dramatic consequences for the prospect of realizing DIQKD experimentally. A detailed feasibility study using a SPDC source is on-going, and first estimates accounting for noise and finite-statistics effects are promising for demonstrating that a key can be distributed over short distances with fully device-independent security guarantees. 

Note that our results are also relevant in a relaxed device-independent scenario where one assumes that a part of the local noise of Alice's and Bob's device is well characterized and cannot be controlled by Eve. Given the amount of such ``trusted'' noise for both parties, $p_A$ and $p_B,$ we can  reconstruct the CHSH value that would be observed without this noise. The latter can then be used to bound Eve's information using Ineq. \eqref{eq:noisybound}, where $p$ can be taken to be equal to $p_B$ (or, by adding extra artificial noise, a larger value). By strengthening  the assumption (that some noise is not under control of Eve) one can thus further improve the efficiency, i.e.\ a positive key rate can be obtained with reduced requirements on the detection efficiency.\\

\paragraph{Acknowledgments--}
We thank Rotem Arnon-Friedman, Enky Oudot and Julian Zivy for enlightening discussions, Valerio Scarani and Charles Ci Wen Lim for comments on the manuscript. We acknowledge funding by the Swiss National Science Foundation (SNSF), through the Grants PP00P2-179109, 200021E-176284 (TriQuI) as well as via the National Center for Competence in Research for Quantum Science and Technology (QSIT). 

\newpage
\onecolumngrid
\appendix

\section{Understanding the symmetrisation step}
\label{app:symm}
We first give a detailed description of the symmetrisation step --- in particular, when the outcome of the measurement $A_0$ is not binary-valued, some care needs to be taken in interpreting this process. For this section, it is convenient to interpret all two-outcome measurements as giving outcomes in $\mathbb{Z}_2$ rather than $\{-1,+1\}$, in contrast to the rest of this work. Addition modulo $2$ will be denoted as $\oplus$. 

To start with, consider the situation where noisy pre-processing is not applied. In that case, we should understand symmetrisation to be carried out as follows: Bob generates a uniform random bit $\rv{T}$ and transforms his measurement output to $\rv{B}'_y = \rv{B}_y \oplus \rv{T}$, then sends $\rv{T}$ to Alice over a public channel. If Alice performed measurement $A_1$ or $A_2$, she transforms her measurement output to $\rv{A}'_x = \rv{A}_x \oplus \rv{T}$ as well (hence the CHSH value is invariant under this procedure); however, if she performed measurement $A_0$, she simply stores the value of the bit $\rv{T}$. In that case, the cost of one-way error correction from Bob to Alice is supposed to be quantified by $H(\rv{B}'_1|\rv{A}_0 \rv{T})$, rather than $H(\rv{B}_1|\rv{A}_0)$. For the purpose of computing this value, however, we note that~\cite{Scarani08} (letting $\rv{T}'$ be a copy of $\rv{T}$)
\begin{align}
H(\rv{B}_1|\rv{A}_0) 
&= H(\rv{B}_1 \rv{A}_0) - H(\rv{A}_0) + H(\rv{T}'\rv{T}) - H(\rv{T}) \\
&= H(\rv{B}_1 \rv{T}' |\rv{A}_0 \rv{T}) \quad \text{since $\rv{T}', \rv{T}$ are independent of $\rv{B}_1, \rv{A}_0$}\\
&= H(\rv{B}'_1 \rv{T}' |\rv{A}_0 \rv{T}) \label{eq:2outof3} \\
&= H(\rv{B}'_1 \rv{A}_0 \rv{T}) - H(\rv{A}_0 \rv{T}) \quad \text{since } H(\rv{B}'_1 \rv{T}' \rv{A}_0 \rv{T}) = H(\rv{B}'_1 \rv{A}_0 \rv{T}) \\
&= H(\rv{B}'_1 |\rv{A}_0 \rv{T}), \label{eq:entsymm}
\end{align}
where line~\eqref{eq:2outof3} holds because knowing any two out of $(\rv{B}_1,\rv{B}'_1,\rv{T}')$ completely determines the third value. Hence to know the value of $H(\rv{B}'_1|\rv{A}_0 \rv{T})$, we can simply compute $H(\rv{B}_1|\rv{A}_0)$ instead (which can be done directly from the known distribution $\pr_{\rv{B}_1 \rv{A}_0}$).

We now consider how noisy pre-processing interacts with the symmetrisation step. Note that noisy pre-processing can be modelled as replacing the symmetrised bit $\rv{B}'_1$ by $\chat{\rv{B}}'_1 = \rv{B}'_1 \oplus \chat{\rv{T}}$, where $\chat{\rv{T}}$ is a biased random bit with $\pr(\rv{T}=1)=p$. Then we have $\chat{\rv{B}}'_1 = \rv{B}_1 \oplus \chat{\rv{T}} \oplus \rv{T}'$, so we see that the operations of symmetrisation and noisy pre-processing commute, and we can perform an analysis based on either ordering. For deriving bounds on $H(\chat{\rv{B}}'_1|E\rv{T})$ in the subsequent sections, it will be convenient to assume symmetrisation is applied before noisy pre-processing. However, when computing the error-correction term $H(\chat{\rv{B}}'_1|\rv{A}_0 \rv{T})$, it is more convenient to assume noisy pre-processing is applied before symmetrisation, since we can then apply the same argument as in the preceding paragraph to show this is equal to $H(\chat{\rv{B}}_1|\rv{A}_0)$, so we can simply compute the latter instead (from the known distribution $\pr_{{\chat{\rv{B}}_1} \rv{A}_0}$).

By exploiting the fact that these operations commute, we can also argue that in fact the symmetrisation step is purely for convenience in the proofs and unnecessary to implement in practice, following the approach in~\cite{Scarani08}. Specifically, the bounds we derive in the subsequent sections are bounds on $H(\chat{\rv{B}}'_1|E\rv{T})$. However, if noisy pre-processing is performed first and the symmetrisation step is omitted, 
then the same calculation as that leading to~\eqref{eq:entsymm} (with $\chat{\rv{B}}_1, \chat{\rv{B}}'_1, E$ in place of $\rv{B}_1, \rv{B}'_1, \rv{A}_0$ respectively) shows that $H(\chat{\rv{B}}_1|E) = H(\chat{\rv{B}}'_1 |E \rv{T}).$
Hence our bounds are also valid for $H(\chat{\rv{B}}_1|E)$, i.e.~the value without implementing the symmetrisation step.


\section{Bell-diagonal reduction}
\label{sec:BelldiagReduction}

The aim of this section is, qualitatively speaking, to prove that we can restrict our analysis to Eve performing a ``classical mixture'' of strategies parametrised by indices $(\lambda,\mu)$, where for each value of $(\lambda,\mu)$, the Alice-Bob-Eve state is of the form
\begin{equation}
\ket{\Psi^{\lambda\mu}}_{A'B'E} = \sum_{i=1}^4 \sqrt{L_i} \ket{\Phi_i}_{A'B'}\ket{i}_E,
\end{equation}
where $\ket{\Phi_1} = \ket{\Phi^+}$, $\ket{\Phi_2} = \ket{\Psi^-}$, $\ket{\Phi_3} = \ket{\Phi^-}$ and $\ket{\Phi_4} = \ket{\Psi^+}$ are the four Bell states, and $\bf L$ is a vector of probabilities that can be taken to satisfy $L_1\geq L_2$ and $L_3\geq L_4$. Additionally, for each $(\lambda,\mu)$ the corresponding measurements can be assumed to be qubit Pauli measurements in the $X$-$Z$ plane of the Bloch sphere.
The proof is very similar to that in Ref.~\cite{Pironio09}, though we present the argument in a slightly different order to make it clear that it still applies even with noisy pre-processing. The calculations performed in this section are all straightforward, but are given in some detail for completeness.

Consider any state and measurements that could be used in the protocol, giving rise to some probability distribution $\pr(a,b|x,y)$. (We will generally use small letters to denote values taken by random variables, and capital letters to denote the random variables (equivalently, classical registers) themselves.)
A symmetrisation step is then implemented on all measurement outcomes via a publicly communicated bit $\rv{T}$, yielding symmetrised probabilities $\pr'(a,b|x,y) = (\pr(a,b|x,y) + \pr(\flip{a},\flip{b}|x,y))/2$, 
where we use the abbreviated notation 
$\flip{c}\defvar -c$ for $c \in \{-1,+1\}$. 
We will use the notation $\rv{B}_1$ for the outcome of Bob's first measurement and $\rv{B}_1'$ for his bit after symmetrization (similarly for Alice). After the symmetrisation step, noisy pre-processing is applied, giving rise to some final conditional entropy $H(\chat{\rv{B}}'_1|E\rv{T})$ (we will use $\chat{\rv{B}}'_1$ to denote the result of the noisy pre-processing applied to $\rv{B}'_1$). Our goal is to find a lower bound on $H(\chat{\rv{B}}'_1|E\rv{T})$ that would hold for all states and measurements producing the outcome probabilities $\pr'(a,b|x,y)$ after symmetrisation.
(In principle, one could instead use the unsymmetrised probabilities $\pr(a,b|x,y)$ as the constraints, but this will not be what we consider in this work.)

As a first step, we observe that the combined process of measuring the state and symmetrising the outcomes can be viewed as Alice and Bob's devices making joint measurements across their quantum systems and their local copies of $\rv{T}$. (Explicitly, if the state was originally $\rho^0_{ABE}$ and the measurements $\proj^0_{a|x}, \proj^0_{b|y}$, we could consider the combined process to be applying the measurements $\proj_{a|x} = \proj^0_{a|x} \otimes \pure{0} + \proj^0_{\flip{a}|x} \otimes \pure{1} , \proj_{b|y} = \proj^0_{b|y} \otimes \pure{0} + \proj^0_{\flip{b}|y} \otimes \pure{1}$ to the state $\rho^0_{ABE} \otimes 
(\pure{000}+\pure{111})/2$.) Since we are interested in bounding the conditional entropy of $\chat{\rv{B}}'_1$ produced from the post-symmetrisation outcomes, it would suffice to find a bound on $H(\chat{\rv{B}}'_1|E)$ that would hold for \emph{any} state $\rho_{ABE}$ and measurements producing the outcome distribution $\pr'(a,b|x,y)$, since in particular this would include the state and measurements that describe the combined measurement-and-symmetrisation process. Therefore, we focus on bounding $H(\chat{\rv{B}}'_1|E)$ for states and measurements directly producing $\pr'(a,b|x,y)$. Also, we can restrict ourselves to pure states without loss of generality, since if the state is mixed we can simply extend $E$ to include its purification without reducing Eve's power, due to the data-processing inequality. The measurements can also be assumed to be projective by performing an appropriate Naimark dilation that preserves both the outcome probabilities and Eve's conditional states; see for instance~\cite{Tan19} for details.

Having reduced to this scenario, now consider any pure state $\ket{\psi}_{ABE}$ and projective measurements $\proj_{a|x},\proj_{b|y}$ producing the outcome distribution $\pr'(a,b|x,y)$.
Since Alice only has two possible measurements and each has two outcomes, by Jordan's lemma~\cite{Pironio09, Scarani12} all her projectors can be written as a direct sum of qubit projectors, which we express by factorising her Hilbert space as $A=A'\otimes A''$ and writing: 
\begin{align}
\proj_{a|x} = \sum_{\lambda} \proj^{\lambda}_{a|x} \otimes \pure{\lambda},
\end{align}
where all $\proj^{\lambda}_{a|x}$ are rank-1 projectors
on a qubit system $A'$, and $\{\ket{\lambda}\}$ forms an orthonormal basis for $A''$. Similarly, for Bob we can write $\proj_{b|y} = \sum_{\mu} \proj^{\mu}_{b|y} \otimes \pure{\mu}$. 
Since $\ket{\lambda},\ket{\mu}$ form bases for $A'',B''$, we can always write
\begin{align}
\ket{\psi}_{ABE} = \sum_{\lambda\mu} \sqrt{p_{\lambda\mu}} \ket{\lambda\mu}_{A''B''} \otimes \ket{\psi^{\lambda\mu}}_{A'B'E},
\end{align}
for some probability distribution ${p_{\lambda\mu}}$. 
If the measurements $(x,y)$ are performed and the outcomes $(a,b)$ are obtained, Eve's (subnormalised) conditional states are easily computed to be\footnote{More abstractly, we could have obtained this result by noting that $\proj_{a|x} = \sum_{\lambda} (\id_{A'} \otimes \pure{\lambda}) \proj_{a|x} (\id_{A'} \otimes \pure{\lambda})$ and analogously for $\proj_{b|y}$, then applying cyclicity of partial traces to argue that Eve's conditional states are the same as if $\ket{\psi}_{ABE}$ is replaced by the state obtained after the projective measurements described by $\pure{\lambda},\pure{\mu}$ are applied to $A'',B''$. Conceptually, this aligns more closely with the spirit of the argument in~\cite{Pironio09}. (However, note that the conditional states of the full $ABE$ system \emph{are} changed if those projective measurements are performed first --- only the reduced states on $E$ are unaffected.)} 
\begin{align}
\rho_{E|a_x b_y} 
&= \sum_{\lambda\mu} p_{\lambda\mu}\, \rho^{\lambda\mu}_{E|a_x b_y}, \text{ where } \rho^{\lambda\mu}_{E|a_x b_y} \defvar \trace_{A'B'} \left[ (\proj^\lambda_{a|x}\otimes\proj^\mu_{b|y}\otimes\id_E) \ket{\psi^{\lambda\mu}}_{A'B'E} \right],
\end{align}
introducing an abbreviated notation $\trace_\alpha[\ket{\Psi}] \defvar \trace_\alpha[\pure{\Psi}]$ for pure states.
All our quantities of interest are completely determined by the states $\rho_{E|a_x b_y}$. Specifically, we have $\pr'(a,b|x,y) = \trace[\rho_{E|a_x b_y}]$, and in the case where Bob performs the key-generating measurement, Eve's states conditioned on the outcome are $\rho_{E|b_1} = \sum_{a_x} \rho_{E|a_x b_1}$ (where any input $x$ for Alice can be chosen, due to no-signalling. 
After noisy pre-processing (which can be modelled as the mapping $\pure{b} \to (1-p)\pure{b} + p\pure{\flip{b}}$), 
the final c-q state can be directly computed as
\begin{align}
\hat{\rho}_{\Bkeyreg E} = \sum_{b_1} \pure{b_1}_{\Bkeyreg} \otimes \left(\sum_{\lambda\mu} p_{\lambda\mu}\, \hat{\rho}^{\lambda\mu}_{E|b_1}\right), \text{ where }
\hat{\rho}^{\lambda\mu}_{E|b_1} \defvar (1-p)\rho^{\lambda\mu}_{E|b_1} + p\rho^{\lambda\mu}_{E|\flip{b_1}} .
\end{align}
Our objective is to bound $H(\Bkey|E)_{\hat{\rho}}$. Note that by concavity of the conditional entropy, we have $\sum_{\lambda\mu} p_{\lambda\mu} {H}(\Bkey|E)_{\hat{\rho}^{\lambda\mu}} \leq {H}(\Bkey|E)_{\hat{\rho}}$, where ${H}(\Bkey|E)_{\hat{\rho}^{\lambda\mu}}$ refers to the value that would be given by the conditional states $\hat{\rho}^{\lambda\mu}_{E|b_1}$. 

We shall now show that if we introduce systems $R,R',R''$ with dimensions $\dim(R)=\dim(R')=2$ and $\dim(R'')=\dim(A'')\dim(B'')$, we can define a state $\sigma_{ABERR'R''}$ in the form\footnote{To be fully precise, Eq.~\eqref{eq:betterstate} should instead state that $\sigma^{\lambda\mu}_{A'B'}$ is Bell-diagonal (also, $\proj^\lambda_{a|x}, \proj^\mu_{b|y}$ describe Pauli measurements in the $X$-$Z$ plane) in \emph{some} choice of local bases, which may depend on the pair $(\lambda,\mu)$. However, since this is already sufficient to derive the qubit entropy bound in Sec.~\ref{app:entropybound}, that bound will apply to each $\ket{\psi^{\lambda\mu}}_{A'B'E}$ term individually, and the rest of our argument carries through. Alternatively, we could construct a new state and projectors in which all $\sigma^{\lambda\mu}_{A'B'}$ are simultaneously Bell-diagonal, though the construction would require giving Alice and Bob a copy of each other's variables $\lambda,\mu$.}
\begin{align}
\sigma = \sum_{\lambda\mu} p_{\lambda\mu} \pure{\lambda\mu}_{A''B''} \otimes \pure{\lambda\mu}_{R''} \otimes \sigma^{\lambda\mu}_{A'B'ERR'} 
\text{\quad where $\sigma^{\lambda\mu}_{A'B'} = \text{Tr}_{ERR'} \sigma^{\lambda\mu}_{A'B'ERR'}$ are Bell-diagonal,}
\label{eq:betterstate}
\end{align}
such that when the measurements $\proj_{a|x},\proj_{b|y}$ are performed on $\sigma$, we get outcome distribution $\pr'(a,b|x,y)$, and after performing noisy pre-processing to get a state $\hat{\sigma}_{\Bkeyreg ERR'R''}$, we have $H(\Bkey|ERR'R'')_{\hat{\sigma}} = \sum_{\lambda\mu} p_{\lambda\mu} {H}(\Bkey|E)_{\hat{\rho}^{\lambda\mu}} \leq {H}(\Bkey|E)_{\hat{\rho}}$. Hence when finding the minimum conditional entropy over all states compatible with $\pr'(a,b|x,y)$, we can restrict ourselves to states in the form~\eqref{eq:betterstate}.

To achieve this, we define
\begin{align}
\begin{gathered}
\sigma^{\lambda\mu}_{A'B'ERR'} \defvar \sum_{rr'} \frac{1}{4} \pure{rr'}_{RR'} \otimes \pure{\phi^{\lambda\mu rr' }}_{A'B'E},\\
\text{with } \ket{\phi^{\lambda\mu 00}} \defvar \ket{\psi^{\lambda\mu}}, \ket{\phi^{\lambda\mu 10}} \defvar U_{YY}\ket{\psi^{\lambda\mu}}, \ket{\phi^{\lambda\mu 01}} \defvar \ket{\psi^{\lambda\mu}}^*,
\ket{\phi^{\lambda\mu 11}} \defvar \left( U_{YY}\ket{\psi^{\lambda\mu}}\right)^*.
\end{gathered}
\label{eq:belldiag}
\end{align}
Here we denote $U_{YY} \defvar \sigma_Y \otimes \sigma_Y \otimes \id_E$, and the Pauli operators $\sigma_Y$ and complex conjugate $^*$ are both defined with respect to a specific basis constructed in~\cite{Pironio09} that ensures the reduced states $\sigma^{\lambda\mu}_{A'B'}$ are Bell-diagonal. To briefly outline the construction: if we pick any local bases for $A',B'$ and define $\sigma_Y$ with respect to those bases, the state $\omega^{\lambda\mu}_{A'B'} \defvar (1/2) \trace_{E}\left[\pure{\phi^{\lambda\mu 00}}_{A'B'E} + \pure{\phi^{\lambda\mu 10}}_{A'B'E}\right]$ will be block-diagonal 
(in $2\times2$ blocks) 
in the induced Bell basis, leaving only two nonzero off-diagonal terms. In particular, there is enough freedom in the basis choice to pick one in which $\proj^\lambda_{a|x}, \proj^\mu_{b|y}$ describe Pauli measurements in the $X$-$Z$ plane \emph{and} the remaining off-diagonal terms of $\omega^{\lambda\mu}_{A'B'}$ are purely imaginary. In that case, we see from definition~\eqref{eq:belldiag} that $\sigma^{\lambda\mu}_{A'B'} = \omega^{\lambda\mu}_{A'B'} + \big(\omega^{\lambda\mu}_{A'B'}\big)^*$ will indeed be Bell-diagonal. Finally, there is still enough freedom in the basis choices to ensure that the Bell-state weights can be ordered such that $L_1\geq L_2$ and $L_3\geq L_4$. 

For the state $\sigma$, Eve's conditional states are straightforwardly computed:
\begin{align}
\begin{gathered}
{\sigma}_{ERR'R''|a_x b_y} = \sum_{\lambda\mu rr' } \frac{1}{4} \, p_{\lambda\mu}\, \pure{\lambda\mu rr'}_{R''RR'} \otimes \sigma^{\lambda\mu rr'}_{E|a_x b_y},\\
\text{where } \sigma^{\lambda\mu rr'}_{E|a_x b_y} \defvar \trace_{A'B'} \left[ (\proj^\lambda_{a|x}\otimes\proj^\mu_{b|y}\otimes\id_{E}) |\phi^{\lambda\mu rr'}\rangle_{A'B'E} \right].
\end{gathered}
\end{align}
Recall that all $\proj^\lambda_{a|x}, \proj^\mu_{b|y}$ describe Pauli measurements in the $X$-$Z$ plane. Therefore, we have $\sigma_Y \proj^\lambda_{a|x} \sigma_Y = \proj^\lambda_{\flip{a}|x}$ and analogously for $\proj^\lambda_{b|y}$. By cyclicity of partial trace, this implies $\sigma^{\lambda\mu 10}_{E|a_x b_y} = \sigma^{\lambda\mu 00}_{E|\flip{a_x} \flip{b_y}}$. 
Additionally, taking the complex conjugate of the state does not affect the outcome probabilities because all matrix elements of $\proj^\lambda_{a|x}, \proj^\mu_{b|y}$ are real. Hence as claimed, the outcome probabilities are indeed equal to
\begin{align}
\sum_{\lambda\mu rr' } \frac{p_{\lambda\mu} }{4} \, \trace \left[\sigma^{\lambda\mu rr'}_{E|a_x b_y}\right] 
= \sum_{\lambda\mu r} \frac{p_{\lambda\mu}}{2} \, \trace \left[\sigma^{\lambda\mu r0}_{E|a_x b_y}\right]
= \frac{\pr'(a,b|x,y) + \pr'(\flip{a},\flip{b}|x,y)}{2}
= \pr'(a,b|x,y).
\end{align}
(Note that the individual terms $\trace \left[\sigma^{\lambda\mu r0}_{E|a_x b_y}\right]$ may not be symmetrised, but after summing over $\lambda\mu$, they yield the distribution $\pr'$, which is symmetrised by definition.) 

Eve's conditional states after noisy pre-processing are 
straightforwardly computed (writing $\sigma^{\lambda\mu rr'}_{E|b_1} = \sum_{a_x} \sigma^{\lambda\mu rr'}_{E|a_x b_1}$):
\begin{align}
\hat{\sigma}_{ERR'R''|b_1} = \sum_{\lambda\mu rr' } \frac{1}{4} \, p_{\lambda\mu}\, \pure{\lambda\mu rr'}_{RR'R''} \otimes \hat{\sigma}^{\lambda\mu rr'}_{E|b_1}, \text{ where } \hat{\sigma}^{\lambda\mu rr'}_{E|b_1} = (1-p)\sigma^{\lambda\mu rr'}_{E|b_1} + p \sigma^{\lambda\mu rr'}_{E|\flip{b_1}}.
\end{align}
Since $RR'R''$ are classical registers, we have $H(\Bkey|ERR'R'')_{\hat{\sigma}} = \sum_{\lambda\mu rr' } (p_{\lambda\mu}/4) {H}(\Bkey|E)_{\hat{\sigma}^{\lambda\mu rr' }}$, where ${H}(\Bkey|E)_{\hat{\sigma}^{\lambda\mu rr' }}$ refers to the value that would be given by the conditional states $\hat{\sigma}^{\lambda\mu rr'}_{E|b_1}$.
Also, since $\sigma^{\lambda\mu 10}_{E|a_x b_y} = \sigma^{\lambda\mu 00}_{E|\flip{a_x} \flip{b_y}}$, we have
\begin{align}
\hat{\sigma}^{\lambda\mu 10}_{E|b_1} 
= (1-p)\sigma^{\lambda\mu 10}_{E|b_1} + p \sigma^{\lambda\mu 10}_{E|\flip{b_1}} 
= (1-p)\sigma^{\lambda\mu 00}_{E|\flip{b_1}} + p \sigma^{\lambda\mu 00}_{E|b_1} 
= \hat{\sigma}^{\lambda\mu 00}_{E|\flip{b_1}}.
\end{align}
Since entropy is invariant under relabelling of basis states, this implies ${H}(\Bkey|E)_{\hat{\sigma}^{\lambda\mu 10}} = {H}(\Bkey|E)_{\hat{\sigma}^{\lambda\mu 00}}$, which in turn equals ${H}(\Bkey|E)_{\hat{\rho}^{\lambda\mu}}$. Additionally, since $H(\Bkey|E)$ can be computed entirely in terms of the eigenvalues of the conditional states, which are invariant under complex conjugation, we have ${H}(\Bkey|E)_{\hat{\sigma}^{\lambda\mu 01}} = {H}(\Bkey|E)_{\hat{\sigma}^{\lambda\mu 00}}$ and ${H}(\Bkey|E)_{\hat{\sigma}^{\lambda\mu 11}} = {H}(\Bkey|E)_{\hat{\sigma}^{\lambda\mu 10}}$. Put together, this implies that $H(\Bkey|ERR'R'')_{\hat{\sigma}} = \sum_{\lambda\mu} p_{\lambda\mu} {H}(\Bkey|E)_{\hat{\sigma}^{\lambda\mu 00}} = \sum_{\lambda\mu} p_{\lambda\mu} {H}(\Bkey|E)_{\hat{\rho}^{\lambda\mu}}$, as claimed

In summary, we can without loss of generality consider a state of the form \eqref{eq:betterstate}. For the calculations in the subsequent section, we will for brevity replace the states $\sigma^{\lambda\mu}_{A'B'ERR'}$ with pure states $\pure{\Psi^{\lambda \mu}}_{A'B'E}$ 
such that the reduced states on $A'B'$ match, i.e.~$\trace_{ERR'}[\sigma^{\lambda\mu}_{A'B'ERR'}] = \trace_{E}[\pure{\Psi^{\lambda \mu}}_{A'B'E}]$. (This can be done without reducing Eve's power, since it is isometrically equivalent to giving Eve purifications of the states $\sigma^{\lambda\mu}_{A'B'}$.) Explicitly, this means we consider a state of the form
\begin{equation}
    \sum_{\lambda\mu} p_{\lambda\mu} \pure{\lambda\mu}_{A''B''} \otimes \pure{\lambda\mu}_{R''} \otimes
    \pure{\Psi^{\lambda \mu}}_{A'B'E},
\end{equation}
where the state $\ket{\Psi^{\lambda \mu}}\in \mathds{C}^2_{A'}\otimes\mathds{C}^2_{B'}\otimes \mathcal{H}_E$  inside each block 
is of the form
\be\label{eq:bell-diagonal}
\ket{\Psi^{\lambda \mu}}= \sum_{i=1}^4 \sqrt{L_i} \ket{i}_E \ket{\Phi_i}_{A'B'},
\ee
with $L_1\geq L_2$ and $L_3\geq L_4$, and where the states $\ket{i}_E$ form an orthonormal basis for a 4-dimensional subspace of $\mathcal{H}_E$. (To be precise, each value of $(\lambda,\mu)$ potentially corresponds to a different Bell basis, weights $L_i$ and orthonormal states $\ket{i}_E$, but for brevity we will not explicitly denote this.)\\

\section{Eve's entropies}
\label{app:entropybound}

In order to obtain the security guarantee for our protocol, we need a lower bound on $$H(\Bkey|E)=H(\Bkey)-(H(E)-H(\Bkey|E))$$ when the symmetrized statistics $\pr'(a,b|x,y)$ are observed (c.f. Appendix A). When symmetrisation is applied we have $H(\Bkey)=H(\rv{B}_1)=1$, so it suffices to find an upper bound on 
\begin{equation}\label{eq: Eves info}
   H(E)-H(E|\Bkey) =H(\rho_E) - \sum_{b_1} p_{b_1} H(\hat{\rho}_{E|b_1}).
\end{equation}
Given the block-diagonal structure of the state $\rho_{ABE}$, the quantity $H(E)-H(E|\Bkey)$ is the weighted sum of the same quantity across all subspaces labeled by $\lambda$ and $\mu$. Hence, the problem can be tackled by upper bounding this entropy difference inside each block (we will explain this in more detail in Sec.~\ref{sec:CHSHbound}). 

We hence focus on a single block specified by $(\lambda,\mu)$. For brevity, we will suppress the specification of the parameters $(\lambda,\mu)$ in this subsection --- it will be understood that we are considering a block specified by some particular value of those parameters.
For any state of the form of Eq.~\eqref{eq:bell-diagonal}, the entropy of Eve's partial state is
\be
H(E) = H\left(\sum_{i=1}^4 L_i \pure{i}\right) = H({\bf L}).
\ee
Bob's measurement $B_1$ is parametrized by a single angle $\phi$, that is
\be
B_1 = \cos(\phi) Z + \sin(\phi)X.
\ee

To obtain Eve's conditional entropy recall that Bob secretly decides to flip his measurement outcome with probability $p<\frac{1}{2}$. Hence, Eve's conditional states are given by the mixtures (from this point on, states are normalised unless otherwise specified)
\be\begin{split}
\hat{\rho}_{E|\pm1} &= 2\, (1-p) \,  \textrm{tr}_{A'B'}\left(\pure{\Psi }\id_{A'E}\otimes\frac{\id_{B'}\pm B_1}{2}\right) +   2\, p \, \textrm{tr}_{A'B'} \left(\pure{\Psi }\id_{A'E}\otimes\frac{\id_{B'}\mp B_1}{2}\right)\\
&= \textrm{tr}_{A'B'} \pure{\Psi }\id_{A'E}\otimes(\id_{B'}\pm (1-2p) B_1).
\end{split}
\ee

One can easily see that $H(\hat{\rho}_{E|+1})=H(\hat{\rho}_{E|-1})$ (c.f. below), hence for the the state $\ket{\Psi}$ one has
\be\label{eq:conditional state}\begin{split}
H(E|B_1) &=H(\hat{\rho}_{E|+1}) \quad \text{with}\\
\hat{\rho}_{E|+1}&= \left(
\begin{array}{cccc}
 L_1 & 0 & \sqrt{L_1} \sqrt{L_3} \sqrt{q} \cos (\phi ) & \sqrt{L_1} \sqrt{L_4} \sqrt{q} \sin (\phi ) \\
 0 & L_2 & \sqrt{L_2} \sqrt{L_3} \sqrt{q} \sin (\phi ) & - \sqrt{L_2} \sqrt{L_4} \sqrt{q} \cos (\phi ) \\
 \sqrt{L_1} \sqrt{L_3} \sqrt{q} \cos (\phi ) & \sqrt{L_2} \sqrt{L_3} \sqrt{q} \sin (\phi ) & L_3 & 0 \\
 \sqrt{L_1} \sqrt{L_4} \sqrt{q} \sin (\phi ) & -\sqrt{L_2} \sqrt{L_4} \sqrt{q} \cos (\phi ) & 0 & L_4 \\
\end{array}
\right)
\end{split}
\ee
where it is convenient to introduce the parameter
\begin{equation}
q = (1-2p)^2 \in (0,1].
\end{equation}

\subsection{Dependence on Eve's conditional entropy on the measurement angle $\phi$}\label{sec:angleDependence}

The entropy of a density matrix is a function of its eigenvalues, and can be obtained from the characteristic polynomial of the matrix $\hat{\rho}_{E|\pm 1}$, which is
\be\label{eq:polynomial}
\begin{split}
P_{{\bf L},q,\phi}(x)= \det(x 
\id -\hat{\rho}_{E|+1}) &= x^4 - x^3 + a_2 x^2 + a_1 x +a_0 \quad \text{with}\\
a_0 &=  L_1 L_2 L_3 L_4 (q-1)^2\\
a_1 &= -\left(L_1 L_2 L_3+L_2 L_4 L_3+L_1 L_2 L_4+L_1 L_3 L_4\right) (1-q)\\
a_2 &= (L_1 L_2+L_3 L_2+L_4 L_2+L_1 L_3+L_1 L_4+L_3 L_4) \\&-\frac{1}{2} \left(L_1+L_2\right) \left(L_3+L_4\right) q -\frac{1}{2} \left(L_1-L_2\right) \left(L_3-L_4\right) q \cos(2 \phi).
\end{split}
\ee

Remarkably, $P(x)$ only depends on the measurement angle $\phi$ via $\cos(2 \phi)$, and so does the entropy $H(E|\rv{B}_1)$. In particular, it implies $H(\hat{\rho}_{E|+1})=H(\hat{\rho}_{E|-1})$, as mentioned previously, since the two states $\hat{\rho}_{E|+1}[\phi] =\hat{\rho}_{E|-1}[\phi+\pi]$ are related by inverting the measurement direction and $\cos(2\phi) = \cos(2\phi +2\pi)$.
Furthermore, one can express Eve's conditional entropy with the help of the variable $C = \cos(2\phi)\in[-1,1]$
\be
H(E|\rv{B}_1)= s({\bf L}, q, C) = H\Big(\text{Roots}\big[P_{{\bf L},q, \text{acos}(C)/2}(x)\big]\Big),
\ee
as the Shannon entropy $H$ of the list of roots of the polynomial $P_{{\bf L},q,\text{acos}(C)/2}(x)$ in Eq.~\eqref{eq:polynomial}. The rest of this section is devoted to the proof of the following statement.

\begin{theorem}
The conditional entropy of Eve $s({\bf L},q,C)$ is a monotonically decreasing function  of $C$.
\end{theorem}
\begin{proof}
In the characteristic polynomial $P_C(x)=P_{{\bf L},q,C}$ of Eq.\eqref{eq:polynomial}, only the coefficient $a_2$ has a dependence on $C$. It will be convenient to express $P_C$ as
\be
P_C(x) = P_0(x) + \xi C x^2 \quad \text{with} \quad \xi = -\frac{1}{2}q(L_1-L_2)(L_3-L_4)\leq 0,
\ee
where $P_0(x)=P_{C=0}(x)$ does not depend on C. Note that for $\xi=0$ the characteristic polynomial does not depend on $C$ and the theorem holds trivially, so that we assume $\xi< 0$ in the following. Since $P_C$ is the characteristic polynomial of a density matrix it has four real roots $p_i\in[0,1]$ -- the eigenvalues of the state with the convention $p_i\geq p_{i+1}$, so
\be\label{eq:polyfactor}
P_C(x) = (x-p_1)(x-p_2)(x-p_3)(x-p_4).
\ee
The remainder of the proof is done in three steps:
\begin{enumerate}
    \item[A.] First, we show that the nontrivial case is where the roots of the polynomial are not degenerate; $p_1>p_2>p_3>p_4$. For all other cases the proof of the theorem is either straightforward, or the values $p_i$ are incompatible with eigenvalues of Eve's state conditional on the considered measurement.
    \item[B.] Given the dependence of the polynomial $P_C(x)$ on $C$ we express the derivatives of the roots $p_i'$ with respect to $C$ as simple functions of $p_i$ and $\xi$. This allows us to express the rate of entropy change as $H'({\bf p})$ in a rather simple form.
    
    \item[C.] By an appropriate change of variables we show that $\textrm{sign}(H'({\bf p}))= \textrm{sign}( G)$, where $G$ is a simple function of three real parameters restricted to some intervals. Finally, we show that $G\leq 0$ on the whole of its domain. This implies $H'({\bf p})\leq 0$ and proves the theorem.
\end{enumerate}

{\bf A.} First, we show that if $C>-1$ any two eigenvalues can not be degenerate unless they are equal to zero. We will prove it by contradiction. Assume that there are two such eigenvalues $p_k=p_{k+1}\neq 0$ for some value $\bar C\in(-1,1)$, the characteristic polynomial reads
\be
P_{\bar C}(x) = (x-p_k)^2(x-p)(x-p').
\ee
The polynomial $P_{\bar C}(x)$ is tangent to the $y=0$ line at $x=p_k$ and has two other intersections with the $y=0$ line at $x=p$ and $x=p'$. It follows that $P_{\bar C(x)}+ \delta x^2$  has at most two real roots for any positive $\delta$. Hence, the perturbed polynomial $P_{\bar C - \delta}(x) = P_{\bar C} - \delta \xi x^2$ only has two real roots and does not correspond to a valid density matrix. This is a contradiction, since for $C>-1$ the value $\bar C - \delta$ corresponds to a valid measurement angle for some positive $\delta$ and should lead to a valid conditional states (with four real eigenvalues). We have thus proved that there are two possibilities
\be
\begin{split}
    &(i)\quad p_1>p_2\geq p_3=p_4=0,\\
    &(ii)\quad p_1>p_2>p_3>p_4.
    \end{split}
\ee\\

The case $(i)$ requires $a_1=a_0=0$ in Eq.~\eqref{eq:polynomial}, which means that at least two of the $L_i$ are zero. Given that $\xi\neq 0$ and the partial ordering of $L_i$ it only leaves one possibilities $L_4=L_2=0$. In this case one can straightforwardly compute the two nonzero eigenvalues of $\hat{\rho}_{E|+1}$ given by
\be
p_\pm = \frac{1}{2} \left(L_1+L_3\pm\sqrt{2 L_3 L_1 q (1+C)+(L_1-L_3)^2}\right).
\ee
From this expression it is obvious that the purity of Eve's conditional state increases with $C$, while its entropy decreases.\\

{\bf B.} The case $ p_1>p_2>p_3>p_4$ is more interesting. We will now study how the eigenvalues $p_i$ change under small perturbations of $C\to C+\delta$ for $-1<C<1.$ By continuity, we know that the four roots $p_i(C)$ of the polynomial $P_C(x)$ in Eq.\eqref{eq:polyfactor} are differentiable functions of $C$. Let us express the roots of the perturbed polynomials $P_{C+\delta}(x)$ as $p_i(C+\delta)=p_i(C)+\Delta^i_\delta$, with $p_i'=\frac{d p_i(C)}{dC}=\lim_{\delta\to 0} \frac{\Delta^i_\delta}{\delta}$. Given the identify $P_{C+\delta}(x)=P_{C}(x)+ \delta \xi x^2$ we have for each i
\be\begin{split}
0=P_{C+\delta}\big(p_i(C+\delta)\big) = P_{C+\delta}\big(p_i&+\Delta^i_\delta\big) = P_C\big(p_i+\Delta^i_\delta\big) +\delta \xi(p_i+\Delta^i_\delta\big)^2\\
\implies\ P_C\big(p_i+\Delta^i_\delta\big)  &= - \delta \xi(p_i+\Delta^i_\delta\big)^2.
\end{split}
\ee
We develop the last expression in first order of $\delta$. Given that $P_C\big(x+\Delta^i_\delta\big) =P_C\big(x) + \frac{dP_C(x)}{dx} \Delta^i_\delta + o(\delta),$ we get for each $i$
\be
\underbrace{P_C(p_i)}_{=0} + \frac{dP_C(x)}{dx}|_{x=p_i}  \Delta^i_\delta = -\delta \xi p_i^2.
\ee
This implies in the limit $\delta\to0$
\be\label{eq:probaprime}
p_i'= -\xi \frac{ p_i^2}{\frac{d P_C(x)}{dx}|_{x=p_i}} =-\xi \frac{p_i^2}{\prod_{k\neq i} (p_i-p_k)}.
\ee\\

Ultimately, we are interested in the entropy $H({\bf p}) = -\sum_i p_i \log(p_i) $ and its susceptibility to variations of $C$.  For convenience, we use a natural logarithm $\log$ instead of $\log_2$ in the definition of the entropy here and until the end of Sec.~\ref{sec:angleDependence}. This does not affect the validity of the following discussion as the two entropies are related by a constant factor. Our aim is to show that Eq.~\eqref{eq:probaprime} implies that $H'({\bf p})=\frac{d H({\bf p})}{dC}$ is negative. Let us first express this quantity
\begin{equation}\nonumber
   -  H'({\bf p}) = \sum_i\left(p_i'\log(p_i)+ p_i \frac{p_i'}{p_i}\right) = \sum_i p_i'\log(p_i) + \sum_i p_i' =\sum_i p_i'\log(p_i).
\end{equation}
Combining the previous expression with Eq.~\eqref{eq:probaprime} gives
\begin{equation}
    H'({\bf p}) = \xi \sum_i \frac{p_i^2\log(p_i)}{\Pi_{k\neq i}(p_i-p_k)}.
\end{equation}
Given that $\xi<0,$ we want to show that $F({\bf p})=\sum_i \frac{p_i^2\log(p_i)}{\Pi_{k\neq i}(p_i-p_k)}$ is positive. To do so, let us change the variables $p_1>p_2>p_3>p_4$ to
\be
\begin{split}
    p_1= p, \quad
    p_2= r_1 p, \quad
    p_3= r_2 p_2 = r_1 r_2 p, \quad \text{and} \quad
    p_4= r_3 p_3 = r_1 r_2 r_3 p,
\end{split}
\ee
with $\frac{1}{4}\leq p \leq 1$, $0 < r_{1,2}< 1$, and $0 \leq r_3 < 1$ (we already covered the case where $p_3=p_4=0$). We express
\begin{align}\nonumber
F({\bf p})&=\sum_i \frac{p_i^2\log(p_i)}{\Pi_{k\neq i}(p_i-p_k)} = \sum_i \frac{p_i^2\log(\frac{p_i}{p}p)}{\Pi_{k\neq i}(p_i-p_k)}= \sum_i \frac{p_i^2(\log(\frac{p_i}{p}) +\log(p))}{\Pi_{k\neq i}(p_i-p_k)}\\
&=\sum_i \frac{p_i^2\log(\frac{p_i}{p})}{\Pi_{k\neq i}(p_i-p_k)} +\log(p) \sum_i \frac{p_i^2}{\Pi_{k\neq i}(p_i-p_k)}.
\nonumber
\end{align}
From Eq.~\eqref{eq:probaprime} it follows that the second term in the last line is zero. Indeed
\be\nonumber
\log(p) \sum_i \frac{p_i^2}{\Pi_{k\neq i}(p_i-p_k)} = \log(p) \frac{\sum p_i'}{-\xi} = -\frac{\log(p)}{\xi} \Big(\underbrace{\sum_i p_i}_{=1}\Big)'= 0.
\ee
Hence, we obtain
\begin{align}\nonumber
&F({\bf p}) = \sum_i \frac{p_i^2\log(\frac{p_i}{p})}{\Pi_{k\neq i}(p_i-p_k)} =  \frac{1}{p}\sum_i \frac{(\frac{p_i}{p})^2\log(\frac{p_i}{p})}{\Pi_{k\neq i}((\frac{p_i}{p})-(\frac{p_i}{p}))}=\\&
\frac{1}{p}\left( 
 \frac{r_1^2 \log(r_1)}{(r_1-1)(r_1 - r_1 r_2)(r_1-r_1 r_2 r_3)}+
  \frac{r_1^2 r_2^2  \log(r_1 r_2)}{(r_1 r_2-1)(r_1 r_2 - r_1)(r_1 r_2-r_1 r_2 r_3)}+
  \frac{r_1^2 r_2^2 r_3^2 \log(r_1 r_2 r_3)}{(r_1 r_2 r_3-1)(r_1 r_2 r_3 - r_1)(r_1 r_2 r_3-r_1 r_2 )}
   \right).
   \nonumber
\end{align}\\

{\bf C.} Straightforward manipulations allow one to rewrite this expression as
\be\begin{split}
F({\bf p}) &= -\frac{(1-r_2)(1-r_2 r_3 )}{p(1-r_1 r_2 r_3)} G({\bf r})\quad \text{with}\\
G({\bf r}) &=\frac{1}{1-r_1 r_2}\left(\frac{\log(r_1)}{1-r_1}- r_2 \frac{\log(r_2)}{1-r_2}\right)  -\frac{r_2r_3}{1-r_2 r_3}\left(\frac{\log(r_2)}{1-r_2}- r_3 \frac{\log(r_3)}{1-r_3}\right).
\end{split}
\ee
Given that the prefactor $-\frac{(1-r_2)(1-r_2 r_3 )}{p(1-r_1 r_2 r_3)} $ is always negative, we want to show that $G({\bf r})\leq 0$ for all $r_1,r_2\in(0,1)$ and $r_3 \in[0,1)$. To do so we change the variables again to 
\be
r = r_2, \quad x= r_1 r_2, \quad \text{and} \quad y= r_3 r_2
\ee
with $r\in(0,1)$, $x \in (0,r)$ and $y \in [0,r)$. With the new variables, $G$ can be expressed as
\be
G(r,x,y) = \frac{r}{r-x}\left(\frac{\log(x)}{1-x} - \frac{\log(r)}{1-r}\right) +
 \frac{y}{r-y}\left(\frac{y \log(y)}{1-y}- \frac{r \log(r)}{1-r}\right).\label{eq:G op}
\ee\\

Let us analyze the two terms separately. For this, it will be useful to use the following bound on the logarithm function, valid for $0<x\leq1$~\cite{Flemming04}:
\be\label{eq:logBounds}
\frac{x^2-1}{2x} \leq \log(x) \leq 2\frac{x-1}{x+1}.
\ee

\begin{figure}[h!]
	\centering
\includegraphics[width=0.8\textwidth]{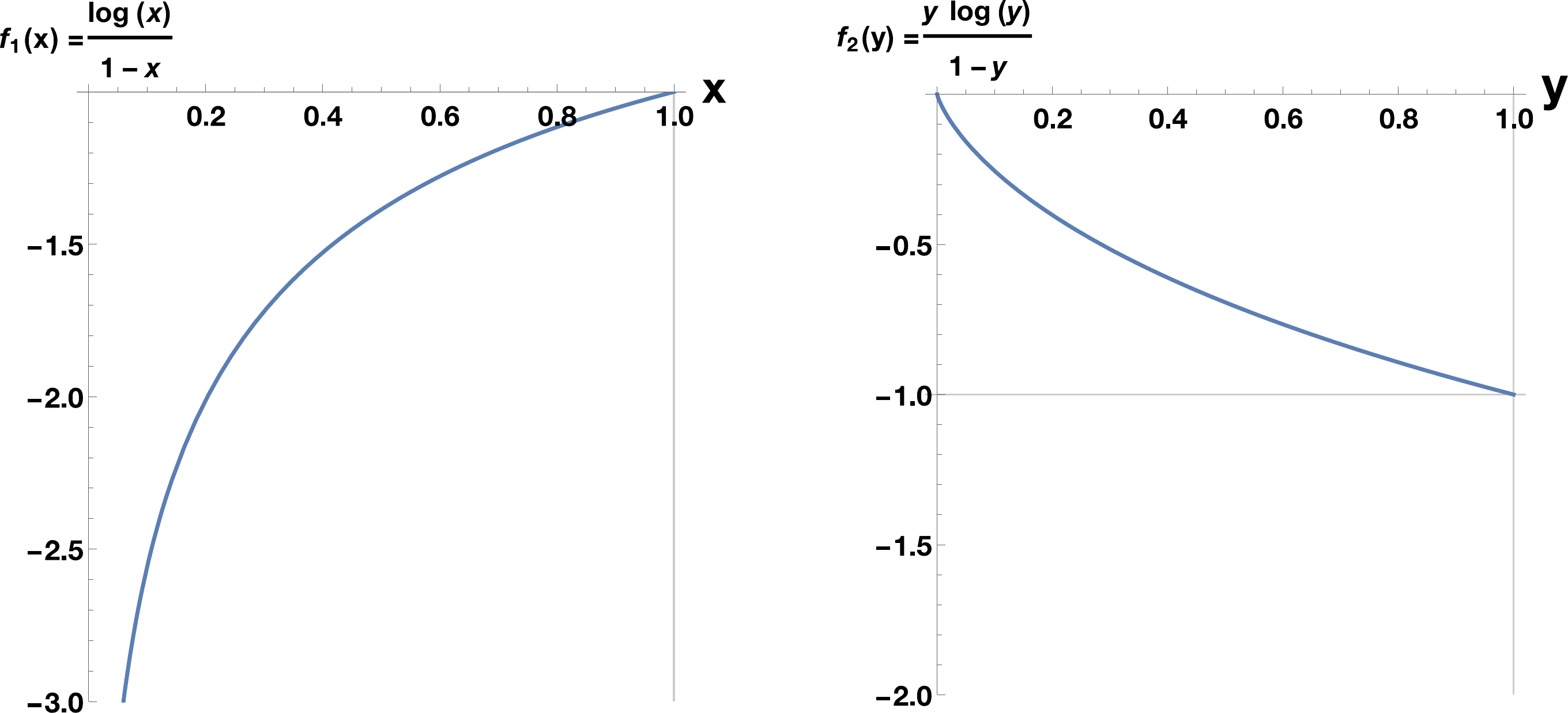}
 \caption{The plots of $f_1(x)=\frac{\log(x)}{1-x}$ and $f_2(y)=\frac{y \log(y)}{1-y}$.}
	\label{fig:plots}
\end{figure}

First, consider the function $f_1(x)=\frac{\log(x)}{1-x}\in(-\infty, -1]$, depicted in the left half of Fig.~\ref{fig:plots}. This function is concave because its second derivative is negative:
\begin{align}
f_1''(x) = - \frac{1-4x + 3 x^2 - 2x^2 \log(x)}{(1-x)^3 x^2} &\leq 0\\
\Leftrightarrow 1-4x + 3 x^2 - 2x^2 \log(x) &\geq 0\\
\Leftarrow 1-4x + 3 x^2 - 2x^2 2\frac{x-1}{x+1} &\geq 0\\
\Leftrightarrow (1-x)^3 & \geq 0,
\end{align}
where we used Eq.~\eqref{eq:logBounds}. Therefore, the ratio $ \frac{f_1(x)-f_1(r)}{r-x} $ is increasing with $x$ ($x<r$) and taking the limit $x\to r$ implies
\begin{equation}\label{eq:f1}
    \frac{f_1(x)-f_1(r)}{r-x} \leq - f_1'(r) = -\frac{1}{(1-r) r}-\frac{\log (r)}{(1-r)^2}.
\end{equation}

Second, we consider the function $f_2(y)=\frac{y \log(y)}{1-y}\in[-1,0]$, depicted in the right half of Fig.~\ref{fig:plots}. Using Eq.~\eqref{eq:logBounds}, we deduce
\begin{align}\nonumber
  f_2''(y) = \frac{-y^2+2 y \log (y)+1}{(1-y)^3 y} & \geq 0 \\ 
  \Leftrightarrow  -y^2+2 y \log (y)+1 & \geq 0  \nonumber \\
  \Leftrightarrow  \log(y)&\geq \frac{y^2-1}{2y} \label{eq:log1}
\end{align}
i.e.~the function $f_2(y)$ is concave.
Therefore, $\frac{f_2(y)-f_2(r)}{r-y}$ is decreasing with $r$ yielding a bound when $y<r$ at the limit $r\to y$ 
\begin{equation}
    \frac{f_2(y)-f_2(r)}{r-y} \leq -f'_2(y) =-\frac{1-y+\log (y)}{(1-y)^2}.
\end{equation}
Then, it follows that 
\begin{equation}
    y \frac{f_2(y) - f_2(r)}{r-y}\leq -y \frac{1-y+\log (y)}{(1-y)^2} =: f_3(y).
\end{equation}
Furthermore, using $\log(y)\leq y-1$ we obtain
\begin{align}\nonumber
  f'_3(y) = -\frac{2(1-y) + (1+y)\log(y)}{(1-y)^3} &\geq 0 \\ 
  \Leftrightarrow  -2(1-y) - (1+y)\log(y) &\geq 0  \nonumber \\
  \Leftrightarrow  \log(y)&\leq 2 \frac{y-1}{1+y} \label{eq:log1}
\end{align}
(see Ineq. \eqref{eq:logBounds} for the last inequality). Therefore $f_3$ is an increasing function and
\begin{equation}\label{eq:f2}
        y \frac{f_2(y)-f_2(r)}{r-y} \leq f_3(y)\leq f_3(r) = -r \frac{1-r+\log(r)}{(1-r)^2}.
\end{equation}

Plugging the two inequalities \eqref{eq:f1} and \eqref{eq:f2} into \eqref{eq:G op} gives
\begin{align}\label{eq:G proof}
    G & =r \frac{f_1(x)-f_1(r)}{r-x} + y \frac{f_2(y)-f_2(r)}{r-y} \\
    & \leq -r f_1'(r) +f_3(r) \\
    & = -\frac{1}{(1-r)}-\frac{r \log (r)}{(1-r)^2} - r \frac{1-r+\log(r)}{(1-r)^2} \\
    & = 
    \frac{r^2-2 r \log (r)-1}{(1-r)^2}\leq r^2-2 r \log (r)-1\leq 0,
\end{align}
where the last inequality has been derived using Ineq.~\eqref{eq:logBounds}.
\end{proof}

\subsection{Eve's information vs CHSH score.}\label{sec:CHSHbound}

Following the logic of the main text, our next step towards the security guarantee is to bound Eve's information, i.e.~to bound Eq.~\eqref{eq: Eves info} from the observed CHSH score $S$ between Alice and Bob. We thus want to find the worst case value of
\be\label{eq: min info}
\begin{split}
I(S) = \max_{\rho_{ABE},B_1}&\ H(E) - H(E|\Bkey)\\
           \text{s.t.}&\ \langle \text{CHSH}\rangle= \textrm{tr} \rho_{ABE} \big(A_1\otimes (B_1 \otimes B_2) + A_2\otimes(B_1 - B_2)\big)\otimes \id_E \geq S.
\end{split}
\ee
Given the block diagonal structure of the state $\rho_{ABE}$ and the measurement operators, both quantities
are averaged over all the subspaces
\be
\begin{split}
H(E) - H(E|\Bkey) &= 
\sum_{\lambda \mu}p_{\lambda \mu} \left( H({\bf L}^{\lambda\mu}) -  H(\hat{\rho}_{E|+1}^{\lambda\mu})\right)\\
\langle \text{CHSH}\rangle &=  \sum_{\lambda \mu} p_{\lambda \mu} \langle\text{CHSH}\rangle_{\lambda \mu}=\sum_{\lambda \mu} p_{\lambda \mu} \bra{\Psi^{\lambda \mu}}
\big(A_1^{\lambda}\otimes (B_1^{\mu} \otimes B_2^{\mu}) + A_2^{\lambda}\otimes(B_1^{\mu} - B_2^{\mu})\big)\otimes \id_E \ket{\Psi^{\lambda \mu}},
\end{split}
\ee
where each state $\ket{\Psi^{\lambda \mu}}$ is the purification of the Bell-diagonal state as given in Eq.~\eqref{eq:bell-diagonal} with the
weights ${\bf L}^{\lambda\mu}$, and Eve's conditional states $\hat{\rho}_{E|+1}^{\lambda\mu}$ have the form of Eq.~\eqref{eq:conditional state}. (We used the equality $H(\hat{\rho}_{E|+1}^{\lambda\mu}) = H(\hat{\rho}_{E|-1}^{\lambda\mu})$ and the fact that the probability for each outcome $\pm 1$ is the same.) It follows that any bound one derives for a restriction to the qubit subspaces
\be
H({\bf L}^{\lambda\mu}) - H(\hat{\rho}_{E|+1}^{\lambda\mu}) \leq I\left(\langle\text{CHSH}\rangle_{\lambda \mu}\right),
\ee
will hold for the overall state if the function $I$ is concave, since we would have
\be
H(E) - H(E|\Bkey) =  \sum_{\lambda \mu}p_{\lambda \mu} \left(H({\bf L}^{\lambda\mu}) - H(\hat{\rho}_{E|+1}^{\lambda\mu})\right) \leq \sum_{\lambda \mu}p_{\lambda \mu}  I\left(\langle\text{CHSH}\rangle_{\lambda \mu}\right) \leq I(\langle\text{CHSH}\rangle).
\ee
Furthermore, the bound is tight if the original inequality is. 

In this section we will derive such a bound. In other words, we will solve the problem of Eq.~\eqref{eq: min info} within a single subspace (from this point forward, we will again suppress explicit specification of the parameters $(\lambda,\mu)$) with a state $\rho_{A'B'E} = \pure{\Psi}$ as given in Eq.~\eqref{eq:bell-diagonal} and Pauli measurements in the the $X$-$Z$ plane (as discussed in Sec.~\ref{sec:BelldiagReduction}, this can be done without loss of generality):
\be
A_x = \cos(\alpha_x) Z + \sin(\alpha_x)X \quad B_y = \cos(\beta_x) Z + \sin(\beta_x)X,
\ee
with $\beta_1=\phi$. \\

Following the notation of \cite{horodecki1995violating} we get
\be
\rho_{AB}=\text{tr}_E \pure{\Psi} =\frac{1}{4} (II\quad XX\quad YY\quad ZZ)\cdot
\left(\begin{array}{c}
L_1+L_2+L_3+L_4\\
L_1-L_2-L_3+L_4\\
-L_1-L_2+L_3+L_4
\\L_1-L_2+L_3-L_4
\end{array}
\right).
\ee
Introducing the unit vectors ${\bf b}_y = \binom{\cos(\beta_y)}{\sin(\beta_y)}$ and ${\bf a}_x = \binom{\cos(\alpha_x)}{\sin(\alpha_x)}$ allows us to express the CHSH score as
\be
\begin{split}
\langle \text{CHSH} \rangle &=  
{\bf b}_1^T \Lambda ({\bf a}_1 + {\bf a}_2)+
{\bf b}_2^T \Lambda({\bf a}_1 - {\bf a}_2)
\\
\text{with}\quad
\Lambda &=\left(
\begin{array}{cc}
  L_1-L_2-L_3+L_4    & \\
     & L_1-L_2+L_3-L_4
\end{array}\right).
\end{split}
\ee
Since the entropies of Eve only depend on the measurement $B_1,$ we can directly maximize the expression above for the remaining measurement directions ${\bf a}_1, {\bf a}_2$ and ${\bf b}_2$. This is easy to do noticing that the expressions $2 \cos(\theta) {\bf c}_1  = ({\bf a}_1 + {\bf a}_2)$ and $2 \sin(\theta) {\bf c}_2  = ({\bf a}_1 - {\bf a}_2)$ define two orthogonal unit vectors ${\bf c}_1$ and ${\bf c}_2$. Hence,
\be\begin{split}
\max_{{\bf a}_1, {\bf a}_2, {\bf b}_2} \langle \text{CHSH} \rangle &= \max_{{\bf c}_j\cdot {\bf c}_i =\delta_{ij}, \theta} 2\cos (\theta) {\bf b}_1^T \Lambda {\bf c}_1  + 2\sin(\theta) ||\Lambda {\bf c}_2|| = 2 \max_{{\bf c}_j\cdot {\bf c}_i =\delta_{ij}, \theta}  \big(\cos(\theta)\quad \sin (\theta)\big)\cdot \binom{{\bf b}_1^T \Lambda {\bf c}_1}{||\Lambda {\bf c}_2||}
\\&= 2 \max_{{\bf c}_j\cdot {\bf c}_i =\delta_{ij}} \sqrt{\left({\bf b}_1^T \Lambda {\bf c}_1\right)^2 +||\Lambda {\bf c}_2||^2}.
\end{split}
\ee
For the first equality, we used the fact that the maximum $\max_{\bf b_2} {\bf b}_2^T \Lambda {\bf c}_2 = ||\Lambda {\bf c}_2||$ is attained for ${\bf b}_2 \parallel \Lambda {\bf c}_2.$ The same argument is used in the last equality for the maximization over the $\binom{\cos(\theta)}{\sin(\theta)}$ vector. The last quantity in the equalities above satisfies
\be
\sqrt{\left({\bf b}_1^T \Lambda {\bf c}_1\right)^2 +||\Lambda {\bf c}_2||^2} \leq \sqrt{||\Lambda {\bf c}_1||^2  +||\Lambda {\bf c}_2||^2} = \sqrt{\text{tr} \Lambda^2},
\ee
which can always be attained for any measurement $B_1$ by setting ${\bf c}_1 \parallel \Lambda {\bf b}_1$. Hence, we find that the maximal CHSH value of 
\be\label{eq:CHSHscore}
\langle \text{CHSH} \rangle = 2\sqrt{\text{tr} \Lambda^2}= 2\sqrt{2} \sqrt{\left(L_1-L_2\right){}^2+\left(L_3-L_4\right){}^2}
\ee
can be attained for any value of the setting of the key generating measurement ${\bf b}_1 =\binom{\cos(\phi)}{\sin(\phi)}$.\\

In addition, we know from Proposition 1 that Eve's conditional entropy is minimal when the value of $\cos(2\phi)$ is maximal, i.e.~$\phi =0$. This corresponds to the worst case maximizing $H(E) - H(E|\Bkey).$ Hence, we get
\be\begin{split}\label{eq: max I}
I(S) = \max_{\bf L}&\ H({\bf L}) - H(\hat{\rho}_{E|+1}|_{\phi=0})\\
\text{s.t.}&\ \langle \text{CHSH} \rangle = 2\sqrt{2}  \sqrt{\left(L_1-L_2\right){}^2+\left(L_3-L_4\right){}^2} \geq S,
\end{split}
\ee
with
\be
\begin{split}
H(\hat{\rho}_{E|+1}|_{\phi=0}) &= H\left(
\left(
\begin{array}{cccc}
 L_1 & 0 & \sqrt{L_1} \sqrt{L_3} \sqrt{q} & 0 \\
 0 & L_2  & 0 & -\sqrt{L_2} \sqrt{L_4} \sqrt{q} \\
 \sqrt{L_1} \sqrt{L_3} \sqrt{q} & 0 & L_3 & 0 \\
 0 & -\sqrt{L_2} \sqrt{L_4} \sqrt{q} & 0 & L_4 \\
\end{array}
\right) \right)=H ({\bf p})
\end{split}
\ee
for
\be
\begin{split}
{\bf p}&=\left(\begin{array}{c}
\frac{1}{2} \left(L_1+L_3 +\sqrt{4 L_1 L_3 q +(L_1-L_3)^2}\right)\\
\frac{1}{2} \left(L_1+L_3 -\sqrt{4 L_1 L_3 q +(L_1-L_3)^2}\right)\\
\frac{1}{2} \left(L_2+L_4 +\sqrt{4 L_2 L_4 q +(L_2-L_4)^2}\right)\\
\frac{1}{2} \left(L_2+L_4 -\sqrt{4 L_2 L_4 q +(L_2-L_4)^2}\right)
\end{array}
\right).
\end{split}
\ee
Given the form of ${\bf p},$ it is convenient to do the following change of variables
\be
\label{eq: P x y}
L_1 = P x, \quad L_3 = P (1-x), \quad L_2 = (1-P)y, \quad \text{and} \quad L_4 = (1-P)(1-y),
\ee
with $x,y,P\in [0,1]$. The partial ordering of the $\bf L$ coefficients implies 
\be\label{eq: ordering}
\begin{cases}
P x \geq (1-P)y\\
P(1-x) \geq (1-P)(1-y)
\end{cases}
\implies
(1-P)y \leq P x \leq (1-P)y + 2P-1
\ee
which requires $P\geq \frac{1}{2}.$ 
The maximization given in Eq.~\eqref{eq: max I} can be rewritten in terms of the new variables as
\begin{align}
I(S) &= \max_{P,x,y} P h_q(x) + (1-P) h_q (y) \\
\text{s.t.} \quad \langle \text{CHSH} \rangle &= 2 \sqrt{2} \sqrt{(P (x+y-2)-y+1)^2+(y-P (x+y))^2} \geq S,
\end{align}
with
\begin{equation}
\begin{split}\label{eq: def hq}
 h_q(z) & = h(z) - h(n_q(z)) \\
 h(z) &= -z \log(z) - (1-z) \log(1-z)  \\ 
 n_q(z) &=\frac{1 + \sqrt{1 - 4\, (1-q)\, z (1-z)}}{2}.\\
  \end{split}
\end{equation}
Here $h(z)$ is the binary entropy function and $\log(z)$ is the logarithm in base 2.

We start by listing some properties of the $h_q$ function that we will prove at the end of the section. First, one easily verifies that it is symmetric around $z=1/2$: $h_q(z)=h_q(1-z).$ Its values at the boundaries are $h_q(0)=h_q(1)=0$. In addition, its derivative $h'_q(z)$ is strictly decreasing with $z$,
that is to say $h''_q(z)<0$. In particular, this implies that $h_q(z)$ is a concave function that increases monotonically for $z\leq 1/2$.

Let us now fix the value of $P$ and consider the curve in the $(x,y)$-plane that corresponds to a constant CHSH score
\be
0 = d \langle \text{CHSH} \rangle =  \frac{d\langle \text{CHSH} \rangle}{dx} dx + \frac{d\langle \text{CHSH} \rangle}{dy} dy  \propto P dx - (1-P) dy.
\ee
The CHSH score remains constant if $P dx = (1-P) dy.$ Hence, we write $dx = (1-P) d\mu$ and $dy = P d\mu$ which parametrize the one-dimensional curve of constant $\langle \text{CHSH} \rangle$ and $P$ in the $x-y$ space with a single parameter $\mu$.  The full curve is given by
\be
(x,y)_\mu = \Big((1-P)\mu + x_0, P \mu \Big) \quad \text{for} \quad \mu \in [0, \frac{1}{P}]
\ee
and for any value $0\leq x_0\leq \frac{2P-1}{P}.$ (Note that the bounds on $\mu$ and $x_0$ come from the constraints on $\bf L$.)
We are interested in maximizing Eve's information $H({\bf L}) - H({\bf p})$ along the curve. To do so we express the infinitesimal variation of this quantity along the curve, we have
\be
d(H({\bf L}) - H({\bf p}) ) = P h_q'(x) dx + (1-P) h_q'(y) = (  h_q'(x) +  h_q'(y)) d\mu.
\ee
Thus, Eve's information has a local extremum for $h_q'(x) = - h_q'(y).$ Given the symmetry of $h_q$ and the fact that  $h'_q$ is strictly decreasing, this equality can only be fulfilled if $x$ and $y$ are symmetric around $\frac{1}{2}$, that is $x+y=1$. Furthermore, the unique extremum is actually a local maximum as the second derivative is negative
\be
\frac{d^2(H({\bf L}) - H({\bf p}) )}{d\mu^2}= \frac{1}{P} h_q''(x)+\frac{1}{1-P} h_q''(y) < 0,
\ee
since $h''_q(z)<0.$ 

By continuity it follows that $x+y =1$ is also the global maximum of $H({\bf L}) - H({\bf p})$ along the $\mu$-curve. Hence without loss of generality we can set $y=1-x$ in our constrained maximization. The functions of interest then become 
\be
    \begin{split}
        I(S)&=\max_{x,P} h_q(x)\\
        \text{s.t.}  \quad \langle \text{CHSH} \rangle &= 2 \sqrt{2-4 (1-P) P-4 (1-x) x} \geq S,
    \end{split}
    \ee
where we used $h_q(x)=h_q(1-x)$.

We want to find the highest possible value of $I$ compatible with the observed value of CHSH. Since $h_q(x)$ is independent of $P,$ we can take the value of $P$ leading to the lowest constraint on $x,$ i.e. $P=1$ (recall that $P\geq\frac{1}{2}$) which maximizes the CHSH score to $\langle \text{CHSH} \rangle=2 \sqrt{2} \sqrt{1-2 x(1-x)}$. This implies 
    \be\label{eq: IEq}
    H({\bf L}) - H({\bf p}) \leq I_E(S;q)= h_q\left(\frac{1}{2}\Big(1+\sqrt{2\left(\frac{S}{2\sqrt{2}}\right)^2-1}\Big)\right).\\
\ee
A sketch of the function $I_E(S;q)$ for different choices of the parameter $q$ is given in Fig.~\ref{fig:hq}. One notices that our whole construction is tight, i.e.~it identifies a strategy for Eve that attains the bound. 

For convenience, we write below the explicit formula for the right hand side of the inequality \eqref{eq: IEq} in terms of the original parameter $p$. Using $q=(1-2p)^2$ and Eq.~\eqref{eq: def hq} we get
\be \label{eq: final}
 \boxed{H({\bf L}) - H({\bf p})\leq I_p(S) = I_E(S;(1-2p)^2) =  h\left( \frac{1+\sqrt{(S/2)^2-1}}{2}\right)  - h\left( \frac{1+\sqrt{1-(1-p) p \left(8-S^2\right)}}{2}\right)} .
\ee

\begin{figure}[h!]
	\centering
\includegraphics[width=0.6\textwidth]{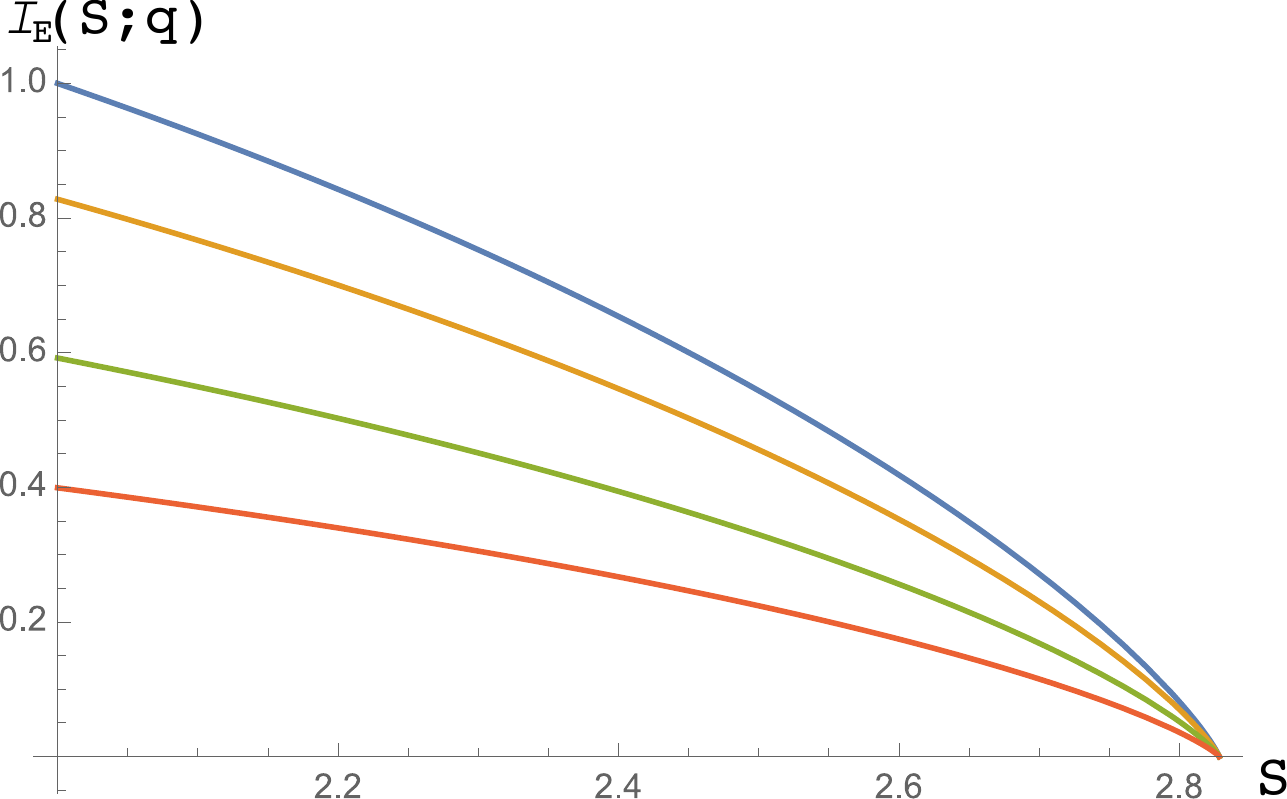}
 \caption{The upper bound $I_E(S;q)$ on Eve's information versus the CHSH score $S$ for noisy-preprocessing values $q=(1-2p)^2=1,0.9, 0.7$ and $ 0.5$ (from top to bottom).}
	\label{fig:hq}
\end{figure}

{\bf Properties of the $h_q$ function.} First, note that $h_{q=0}(z)=h(z)-h(z)=0$, which also follows from the fact that for $q=0$ Bob replaces his outcome by a private random bit. The symmetry of the function $h_q(z)=h_q(1-z)$ trivially follows from the symmetry of the binary entropy $h(z)= h(1-z)$ and of the function $n_q(z)=n_q(1-z)$. Its values at the boundaries are given by
\be
h_q(0)=h_q(1)= h(0) - h(0)=0.
\ee
Similarly, one can obtain the derivative at the boundaries. For $q>0$ one has
\be\begin{split}
h'_q(0) &= - h_q'(1) = \lim_{z\to0_+} h'(z) - h'\big(n_q(z)\big) n_q'(z) = \lim_{z\to0_+} -\log(z) + \log\big((1-q)z\big)(1-q)\\
&= \lim_{z\to0_+} -\log(z) + \log(z^{1-q}) +(1-q)\log(1-q)=  \lim_{z\to0_+} -\log(z^q) +(1-q)\log(1-q) = \infty.
\end{split}
\ee
Finally, for the second derivative we have
\be
\begin{split}
h''_q(z) &=h''(z) -h''(n_q(z)) (n_q'(z))^2 - h'(n_q(z))n_q''(z)   \\
&= -\frac{1}{z(1-z)}- \frac{1}{(1-q) z (1-z)}\frac{(1-q)^2 (1-2 z)^2}{1-4 (1-q) z (1-z)}+\frac{ \tanh ^{-1}\left(\sqrt{1-4 (1-q) z (1-z)}\right) 4 (1-q) q}{(1-4 (q+1) z (1-z))^{3/2}}\\
&=-\frac{q}{(1-z) z (1-4 (1-q) z (1-z))} + \frac{ \tanh ^{-1}\left(\sqrt{1-4 (1-q) z (1-z)}\right) 4 (1-q) q}{(1-4 (q+1) z (1-z))^{3/2}}.
\end{split}
\ee
Introducing $0\leq R(q,z) = 4 (1-q) z (1-z) < 1$ (here $R<1$ follows from $q>0$) the expression simplifies to
\be
h''_q(z)=4 q(1-q) \left(-\frac{1}{R (1-R)}+ \frac{\tanh^{-1}(\sqrt{1-R})}{(1-R)^{3/2}} \right)=
\frac{4 q(1-q)}{(1-R)^{3/2}} \left(-\frac{\sqrt{1-R}}{R}+ \tanh^{-1}\left(\sqrt{1-R}\right)\right).
\ee
The second derivative $h''_q(z)$ is thus negative if  $\frac{\sqrt{1-R}}{R}\geq \tanh^{-1}\left(\sqrt{1-R}\right)$. Replacing $\zeta=\sqrt{1-R}$ and using $\tanh^{-1}(\zeta) = \frac{1}{2}\log\left(\frac{1+\zeta}{1-\zeta}\right)$ we rewrite this inequality as
\be
\frac{2 \zeta}{1-\zeta^2}\geq \log\left(\frac{1+\zeta}{1-\zeta}\right).
\ee
This is easy to prove by noting that the two terms are equal at $\zeta=0$:
\be
\frac{2 \zeta}{1-\zeta^2}|_{\zeta=0} = \log\left(\frac{1+\zeta}{1-\zeta}\right)|_{\zeta=0}=0,
\ee
while the left hand side term grows faster
\be
\frac{d}{d\zeta} \left(\frac{2 \zeta}{1-\zeta^2}\right) = 2 \frac{1+\zeta^2}{(1-\zeta^2)^2}\geq 2 \frac{1}{1-\zeta^2} =\frac{d}{d\zeta} \log\left(\frac{1+\zeta}{1-\zeta}\right) , 
\ee
with equality only for $\zeta=0$. Hence, we have shown that $h''(z)\leq 0$, with $h''(z)= 0$ only possible for $\zeta=0$. However, $\zeta=0$ corresponds to $R=1$ and is impossible for $q>0$. We have thus proved that
\be
h''_q(z)<0.
\ee


\section{Statistics obtained with a SPDC source and photon counting}

\subsection{Modelization of the source and detection devices}
\paragraph{Source --} We consider a source generating entangled photon pairs according to 
\begin{equation}
\nonumber
|\psi\rangle = (1-T_g^2)^{N/2} (1-T_{\bar{g}}^2)^{N/2} \Pi_{k=1}^N e^{T_g a_k^\dagger b_{k,\perp}^\dagger - T_{\bar{g}} a_{k,\perp}^\dagger b_k^\dagger} |\underbar{0}\rangle.
\end{equation}
This describes a bi-partite state where the first party, Alice, measures the modes labelled by the bosonic operators $a_k$ and $a_{k,\perp}$ and the second party, Bob, measures $b_k$ and $b_{k,\perp}.$ In the mono-mode case where $k$ takes a single value, entanglement appears clearly between the modes labelled with and without the index $\perp.$ $k$ is here used to account for multiple emission modes. $T_g=\tanh g$ and similarly $T_{\bar g}=\tanh \bar g$ where $g$ and $\bar{g}$ indicate squeezing parameters related to $a_k^\dagger b_{k,\perp}^\dagger$ and $a_{k,\perp}^\dagger b_k^\dagger$ respectively, can be changed to tune the amount of entanglement per pair. Below, we first consider the mono-mode case (the index $k$ is omitted), compute the statistics and show how it can be extended to the multi-mode case at the end.\\

\paragraph{Detector --}  We consider photon detectors which do not resolve the photon number. The event "no-click" is modelled by 
\begin{equation}
\label{detector}
\hat D_a(\eta_a) = (1-p_{\text{dc}}^a)(1-\eta_a)^{a^\dag a}
\end{equation}
where $a$ specifies the mode which is detected. $\eta_a$ is the detection efficiency (overall detection efficiency including all the lost from the source to the detector) and $p_{\text{dc}}^a$ the dark count probability.\\

\paragraph{Setting choice --} Rotations are possibly performed before the detections so that the photons can be measured in several basis. The detected modes are called $A$ and $A_\bot$ for Alice and are related to the emission modes by 
\begin{eqnarray}
&& a = C_\alpha A + S_\alpha e^{i \phi_\alpha} A_\bot\\
&& a_\bot = S_\alpha e^{-i \phi_\alpha} A - C_\alpha A_\bot
\end{eqnarray}
with $C_\alpha = \cos(\alpha)$ and $S_\alpha=  \, \sin(\alpha)$ and similarly for Bob. \\

\paragraph{Summary}The state which is effectively measured can be written as 
\begin{equation}
\label{rotated_state}
|\psi_{\alpha, \phi_\alpha, \beta, \phi_\beta}\rangle=\left(1-T_g^2\right)^{\frac{1}{2}} \left(1-T_{\bar g}^2\right)^{\frac{1}{2}} e^{(A^\dag, A_{\bot}^\dag) M \left(\begin{matrix}B^\dag \\ B_{\bot}^\dag\end{matrix}\right)} |\underbar{0}\rangle
\end{equation}
with
\begin{equation}
M=\left( \begin{matrix}
T_g C_\alpha S_\beta e^{-i \phi_\beta} - T_{\bar g} S_\alpha e^{-i\phi_\alpha} C_\beta & 
-T_g C_\alpha C_\beta - T_{\bar g} S_\alpha e^{-i \phi_\alpha} S_\beta e^{i \phi_\beta} \\ 
T_g S_\alpha e^{i \phi_\alpha} S_\beta e^{-i \phi_\beta } + T_{\bar g} C_\alpha C_\beta & 
-T_g S_\alpha e^{i \phi_\alpha} C_\beta + T_{\bar g} C_\alpha S_\beta e^{i \phi_\beta} 
\end{matrix}\right).
\end{equation}
It is measured according to a model where the POVM element associated to a no-click in detector $A$ is given by
\begin{equation}
\label{detector}
\hat D_A(\eta_A) = (1-p_{\text{dc}}^A)(1-\eta_A)^{A^\dag A}
\end{equation}
and similarly for $A_\bot,$ $B$ and $B_\bot.$ 

\subsection{Derivation of the full statistics}
The no-click events can all be computed in a similar way. Consider for example the probability $p(\text{nc}_A)$ of having a no-click in A
\begin{equation}
p(\text{nc}_A)=(1-p_{\text{dc}}^A)  \text{Tr}\left(R_A^{A^\dag A/2}|\psi_{\alpha, \phi_\alpha, \beta, \phi_\beta}\rangle \langle \psi_{\alpha, \phi_\alpha, \beta, \phi_\beta}| R_A^{A^\dag A/2}\right)
\end{equation}
with $R_A=1-\eta_A.$ From $x^{a\dag a} f(a^\dag) = f(x a^\dag) x^{a^\dag a},$ we have 
\begin{equation}
R_A^{A^\dag A/2}|\psi_{\alpha, \phi_\alpha, \beta, \phi_\beta}\rangle = \left(1-T_g^2\right)^{\frac{1}{2}} \left(1-T_{\bar g}^2\right)^{\frac{1}{2}} e^{(A^\dag, A_{\bot}^\dag) M_{R_A} \left(\begin{matrix}B^\dag \\ B_{\bot}^\dag\end{matrix}\right)} |\underbar{0}\rangle
\end{equation}
with
\begin{equation}
M_{R_A}=\left( \begin{matrix}
R_A^{1/2} (T_g C_\alpha S_\beta e^{-i \phi_\beta} - T_{\bar g} S_\alpha e^{-i\phi_\alpha} C_\beta) & 
R_A^{1/2} (-T_g C_\alpha C_\beta - T_{\bar g} S_\alpha e^{-i \phi_\alpha} S_\beta e^{i \phi_\beta}) \\ 
T_g S_\alpha e^{i \phi_\alpha} S_\beta e^{-i \phi_\beta } + T_{\bar g} C_\alpha C_\beta & 
-T_g S_\alpha e^{i \phi_\alpha} C_\beta + T_{\bar g} C_\alpha S_\beta e^{i \phi_\beta}
\end{matrix}\right).
\end{equation}
Let $\lambda_1^{R_A}$ and $\lambda_2^{R_A}$ be the singular values of $M_{R_A}.$ We find
\begin{equation}
\text{Tr}\left(R_A^{A^\dag A/2}|\psi_{\alpha, \phi_\alpha, \beta, \phi_\beta}\rangle \langle \psi_{\alpha, \phi_\alpha, \beta, \phi_\beta}| R_A^{A^\dag A/2}\right)=\frac{\left(1-T_g^2\right) \left(1-T_{\bar g}^2\right)}{(1-(\lambda_1^{R_A})^2)(1-(\lambda_2^{R_A})^2)}.
\end{equation} 
We deduce 
\begin{equation}
p(\text{nc}_A)=(1-p_{\text{dc}}^A)\frac{\left(1-T_g^2\right) \left(1-T_{\bar g}^2\right)}{(1-(\lambda_1^{R_A})^2)(1-(\lambda_2^{R_A})^2)}.
\end{equation}
In the multimode case where the source emits $N$ independent modes, we simply have 
\begin{equation}
p(\text{nc}_A)=(1-p_{\text{dc}}^A) \left(\frac{\left(1-T_g^2\right) \left(1-T_{\bar g}^2\right)}{(1-(\lambda_1^{R_A})^2)(1-(\lambda_2^{R_A})^2)}\right)^N.
\end{equation}
The full statistics can be derived in a similar way. In particular, by introducing
\begin{equation}
M_{R_A,R_{A_\bot},R_B,R_{B_\bot}}=\left( \begin{matrix}
R_A^{1/2} R_B^{1/2} (T_g C_\alpha S_\beta e^{-i \phi_\beta} - T_{\bar g} S_\alpha e^{-i\phi_\alpha} C_\beta) & 
R_A^{1/2} R_{B_{\bot}}^{1/2} (-T_g C_\alpha C_\beta - T_{\bar g} S_\alpha e^{-i \phi_\alpha} S_\beta e^{i \phi_\beta}) \\ 
R_{A_{\bot}}^{1/2} R_B^{1/2} (T_g S_\alpha e^{i \phi_\alpha} S_\beta e^{-i \phi_\beta } + T_{\bar g} C_\alpha C_\beta) & 
R_{A_{\bot}}^{1/2} R_{B_{\bot}}^{1/2} (-T_g S_\alpha e^{i \phi_\alpha} C_\beta + T_{\bar g} C_\alpha S_\beta e^{i \phi_\beta})
\end{matrix}\right)
\end{equation}
and the corresponding singular values $\lambda_1^{R_A,R_{A_\bot},R_B,R_{B_\bot}}$ and $\lambda_2^{R_A,R_{A_\bot},R_B,R_{B_\bot}}.$ We find 
\begin{eqnarray}
\nonumber
&&p(\text{nc}_A)=(1-p_{\text{dc}}^A) \left(\frac{\left(1-T_g^2\right) \left(1-T_{\bar g}^2\right)}{(1-(\lambda_1^{R_A,1,1,1})^2)(1-(\lambda_2^{R_A,1,1,1})^2)}\right)^N\\
\nonumber
&&p(\text{nc}_{A_{\bot}})=(1-p_{\text{dc}}^{A_{\bot}}) \left(\frac{\left(1-T_g^2\right) \left(1-T_{\bar g}^2\right)}{(1-(\lambda_1^{1,R_{A_{\bot}},1,1})^2)(1-(\lambda_2^{1,R_{A_{\bot}},1,1})^2)}\right)^N\\
\nonumber
&&p(\text{nc}_B)=(1-p_{\text{dc}}^B) \left(\frac{\left(1-T_g^2\right) \left(1-T_{\bar g}^2\right)}{(1-(\lambda_1^{1,1,R_B,1})^2)(1-(\lambda_2^{1,1,R_B,1})^2)}\right)^N\\
\nonumber
&&p(\text{nc}_{B_{\bot}})=(1-p_{\text{dc}}^{B_{\bot}}) \left(\frac{\left(1-T_g^2\right) \left(1-T_{\bar g}^2\right)}{(1-(\lambda_1^{1,1,1,R_{B_{\bot}}})^2)(1-(\lambda_2^{1,1,1,R_{B_{\bot}}})^2)}\right)^N.
\end{eqnarray}
The same line of thought applies to the joint probability 
\begin{eqnarray}
\nonumber
&& p(\text{nc}_A \& \text{nc}_B)=(1-p_{\text{dc}}^A)(1-p_{\text{dc}}^B) \left(\frac{\left(1-T_g^2\right) \left(1-T_{\bar g}^2\right)}{(1-(\lambda_1^{R_A,1,R_B,1})^2)(1-(\lambda_2^{R_A,1,R_B,1})^2)}\right)^N\\
\nonumber
&& ...
\end{eqnarray}
that is, the joint probability 
\begin{eqnarray}
\nonumber
&& p(\text{nc}_i \& \text{nc}_j)=(1-p_{\text{dc}}^i)(1-p_{\text{dc}}^j) \left(\frac{\left(1-T_g^2\right) \left(1-T_{\bar g}^2\right)}{(1-(\lambda_1^{r^A_{ij},r^{A_{\bot}}_{ij},r^B_{ij},r^{B_{\bot}}_{ij}})^2)(1-(\lambda_2^{r^A_{ij},r^{A_{\bot}}_{ij},r^B_{ij},r^{B_{\bot}}_{ij}})^2)
}\right)^N
\end{eqnarray}
$i,j$ standing for $A,A_{\bot},B$ or $B_{\bot}$ and
\begin{align}
r^\ell_{ij} & = R_\ell \quad \text{if} \quad i \, \text{or} \, j \, \text{equals} \, \ell  \\
&=1 \quad \text{otherwise}.
\end{align}
Similarly
\begin{eqnarray}
\nonumber
&& p(\text{nc}_i \& \text{nc}_j \& \text{nc}_k)=(1-p_{\text{dc}}^i)(1-p_{\text{dc}}^j)(1-p_{\text{dc}}^k) \left(\frac{\left(1-T_g^2\right) \left(1-T_{\bar g}^2\right)}{(1-(\lambda_1^{r^A_{ijk},r^{A_{\bot}}_{ijk},r^B_{ijk},r^{B_{\bot}}_{ijk}})^2)(1-(\lambda_2^{r^A_{ijk},r^{A_{\bot}}_{ijk},r^B_{ijk},r^{B_{\bot}}_{ijk}})^2)
}\right)^N
\end{eqnarray}
$i,j,k$ standing for $A,A_{\bot},B$ or $B_{\bot}$ and
\begin{align}
r^\ell_{ijk} & = R_\ell \quad \text{if} \quad i \, \text{or} \, j \, \text{or} \, k \, \text{equals} \, \ell  \\
&=1 \quad \text{otherwise}.
\end{align}
Finally 
\begin{equation}
\nonumber
p(\text{nc}_A \& \text{nc}_{A_\bot} \& \text{nc}_B \& \text{nc}_{B_\bot}) = (1-p_{\text{dc}}^A)(1-p_{\text{dc}}^{A_\bot})(1-p_{\text{dc}}^B)(1-p_{\text{dc}}^{B_\bot}) \left(\frac{\left(1-T_g^2\right) \left(1-T_{\bar g}^2\right)}{(1-(\lambda_1^{R_A,R_{A_\bot},R_B,R_{B_\bot}})^2)(1-(\lambda_2^{R_A,R_{A_\bot},R_B,R_{B_\bot}})^2)}\right)^N\\
\end{equation}

\subsection{Derivation of the CHSH value}
To compute the CHSH value, the parties need to bin their results, i.e.~they have to decide how to group their four possible events to get outcome $\pm 1$. Alice (Bob) chooses to assign the value $-1$ when the detectors A (B) clicks whereas $A_\perp$ ($B_\perp$) does not and $+1$ to the three remaining events. Hence, $p(-1,-1|x,y) = \text{Tr}\big((\mathbf{1}-\hat D_A(\eta_A)) \hat D_{A_\perp}(\eta_{A,\perp}) ( \mathbf{1}-\hat D_B(\eta_B)) \hat D_{B_\perp}(\eta_{B,\perp}) |\psi_{\alpha, \phi_\alpha, \beta, \phi_\beta}\rangle \langle \psi_{\alpha, \phi_\alpha, \beta, \phi_\beta}|\big)$, which can be related to the no-detection probabilities derived previously via $p(nc_{A_\perp} \& nc_{B_\perp}) - p(nc_{A} \& nc_{A_\perp} \& nc_{B_\perp}) -p(nc_{A_\perp} \& nc_{B} \& nc_{B_{\perp}})+ p(nc_{A} \& nc_{A_\perp} \& nc_{B} \& nc_{B_\perp}).$ Processing $p(+1,-1|x,y),$ $p(-1,+1|x,y)$ and $p(+1,+1|x,y)$ in a similar way gives an expression for the CHSH value as a function of the squeezing parameters $g$ - $\bar g$, the number of modes $N,$ the measurement settings, the detection efficiencies and dark count probabilities.

\subsection{Comparison with an ideal source}
As discussed in the main text, we computed the requirement on the detection efficiency using various security proofs in a setup using polarisation entanglement produced by means of a SPDC source. It is instructive to compare this requirement when using an ideal source producing two-qubit entangled states. \\

For the protocol proposed in Ref.~\cite{Pironio09}, the requirement on the global detection efficiency is $92.7 \%$ when using a SPDC source, which is $3.4\%$ higher than if one used a perfect two-qubit source. Conversely, when using the security proof that is proposed in this manuscript that employs noisy pre-processing with error correction that uses the four-valued outcome $\rv{A}_0$ instead of its binarization as proposed in Ref.~\cite{Ma12}, the global detection efficiency when a SPDC source is used is $83.2\%$, which is only $0.4 \%$ higher compared to the case of a perfect qubit source. We conclude that developing a photon pair source producing exactly two-qubit entangled states is not a priority when considering the first photonic implementation of device-independent quantum key distribution.

\bibliography{references}

\end{document}